\RequirePackage{fix-cm}
\documentclass[smallcondensed]{svjour3}     

\usepackage{amsthm}
\usepackage{thmtools, thm-restate}
\smartqed  
\usepackage{graphicx}
\usepackage[utf8]{inputenc}
\usepackage{verbatim}
\usepackage{adjustbox}
\usepackage{placeins}
\usepackage{amssymb,amsmath,graphicx}
\usepackage{tikz}

\usetikzlibrary{patterns,arrows,decorations.pathreplacing}
\newtheorem*{theorem*}{Theorem}
\newtheorem*{lemma*}{Lemma}
\newtheorem*{conj*}{Conjecture}

\newcommand{\eleq}[1]{\ensuremath{\stackrel{#1}{\sim}}}

\newcommand{\eeq}[1]{\ensuremath{\stackrel{#1}{=}}}
\DeclareMathOperator{\rank}{rank}
\DeclareMathOperator{\supp}{supp}

\begin{document}

\title{On metric regularity of Reed-Muller codes
}


\author{Alexey Oblaukhov 
}


\institute{A. Oblaukhov \at
    Sobolev Institute of Mathematics, \\
    Novosibirsk, Russia \\
    \email{oblaukhov@gmail.com} \\
	The work is supported by Mathematical Center in Akademgorodok under agreement No. 075-15-2019-1613 with the Ministry of Science and Higher Education of the Russian Federation and Laboratory of Cryptography JetBrains Research.     
}

\date{Received: date / Accepted: date}

\maketitle

\begin{abstract}
 In this work we study metric properties of the well-known family of binary Reed-Muller codes. Let $A$ be an arbitrary subset of the Boolean cube, and $\widehat{A}$ be the metric complement of $A$ --- the set of all vectors of the Boolean cube at the maximal possible distance from $A$. If the metric complement of $\widehat{A}$ coincides with $A$, then the set $A$ is called a {\it metrically regular set}. The problem of investigating metrically regular sets appeared when studying {\it bent functions}, which have important applications in cryptography and coding theory and are also one of the earliest examples of a metrically regular set. In this work we describe metric complements and establish the metric regularity of the codes $\mathcal{RM}(0,m)$ and $\mathcal{RM}(k,m)$ for $k \geqslant m-3$. Additionally, the metric regularity of the codes $\mathcal{RM}(1,5)$ and $\mathcal{RM}(2,6)$ is proved. Combined with previous results by Tokareva N. (2012) concerning duality of affine and bent functions, this establishes the metric regularity of most Reed-Muller codes with known covering radius. It is conjectured that all Reed-Muller codes are metrically regular.

\end{abstract}

\pagebreak

\section{Introduction}
\label{intro}

The problem of investigating and classifying {\it metrically regular sets} was posed by Tokareva \cite{TOK12,TOK15} when studying metric properties of {\it bent functions} \cite{ROT76}.
A Boolean function $f$ in even number of variables $m$ is called a {\it bent function} if it is at the maximal possible distance from the set of affine functions.

Bent functions have various applications in cryptography, coding theory and combinatorics \cite{MES16,TOK15}. In cryptography, bent functions are valued because of their outstanding nonlinearity, which allows one to construct S-boxes for block ciphers which possess high resistance to the linear cryptanalysis \cite{MES16}. However, many problems related to bent functions remain unsolved; in particular, the gap between the best known lower and upper bound on the number of bent functions is extremely large; currently known constructions of bent functions are rather scarse. In 2010 \cite{TOK10}, Tokareva has proved that, like bent functions are maximally distant from affine functions, affine functions are at the maximal possible distance from bent functions, thus establishing the {\it metric regularity} of both sets. This discovery arouses interest in studying the property of metric regularity in order to better understand the structure of the set of bent functions.

From the coding theory standpoint, bent functions form the set of points at the maximal possible distance from the Reed-Muller code of the first order in an even number of variables. Therefore, the aforementioned result by Tokareva establishes the metric regularity of the codes $\mathcal{RM}(1,m)$ for even $m$. Reed-Muller codes are extensively studied for many years, but their metric properties, like the covering radius, are very elusive and are being discovered to this day; just recently, Wang has found the covering radius of the code $\mathcal{RM}(2,7)$ to be equal to $40$ \cite{WAN19}. These problems put Reed-Muller codes in our focus of the research of metric regularity.

Let us briefly overview the results obtained in this area. As mentioned before, Tokareva \cite{TOK10} has established the metric regularity of the sets of affine/bent functions. Metric regularity of several classes of {\it partition set functions} is studied in \cite{STA18}, while the works \cite{KOL17,KUT19} touch upon metric properties of certain subclasses of bent functions. Metric regularity has been actively investigated by the author: metric complements of linear subspaces of the Boolean cube are studied in the paper \cite{OBL16}, while the works \cite{OBL18} and \cite{OBL19} are studying possible sizes of the largest and smallest metrically regular set.

In this work we investigate the metric regularity of Reed-Muller codes. Naturally, the knowledge of the covering radius of the code is necessary for working with the set of its most distant points. Among the codes of high order, covering radii of the codes $\mathcal{RM}(k,m)$ for $k \geqslant m-3$ are known. The covering radius of $\mathcal{RM}(1,m)$ for odd $m > 7$ is unknown, but has been determined for $\mathcal{RM}(1,5)$ \cite{BW72} and $\mathcal{RM}(1,7)$ \cite{MYK80,HOU96}. In \cite{SCH81}, Schatz has found the covering radius of $\mathcal{RM}(2,6)$, while recently Wang has established the covering radius of $\mathcal{RM}(2,7)$ \cite{WAN19}. For $m > 7$, the covering radius of $\mathcal{RM}(2,m)$ is currently unknown. We prove that the codes $\mathcal{RM}(k,m)$, for $k = 0$ and $k \geqslant m-3$ and the codes $\mathcal{RM}(1,5)$, $\mathcal{RM}(2,6)$ are metrically regular and also describe their metric complements in most cases.

The paper is structured as follows. After providing necessary definitions and examples, we prove the metric regularity of the code $\mathcal{RM}(1,5)$. After that we establish the metric regularity of the Reed-Muller codes of order $0$, order $m-2$ and higher, and then we move onto the codes of order $m-3$. In order to handle this case, we describe the method of ``syndrome matrices'' of calculating distances from vectors to the punctured $\mathcal{RM}(m-3,m)$ code, based on the ``Covering codes'' \cite{COH97} book by Cohen et al. Following the book, we calculate the covering radius of the Reed-Muller code of order $m-3$, and utilizing the method further, we obtain the metric complement of this code. The description of the complement allows us to establish that only the functions from $\mathcal{RM}(m-3,m)$ are contained in the second metric complement, which proves the metric regularity of the Reed-Muller codes of order $m-3$. We then proceed to establish the metric regularity of the code $\mathcal{RM}(2,6)$, based on the results obtained for the codes $\mathcal{RM}(2,5)$ and $\mathcal{RM}(1,5)$, since the former can be constructed from the latter using the $(\mathbf{u},\mathbf{u+v})$ construction. The paper concludes with an overview of the results obtained and a hypothesis concerning the metric regularity of all Reed-Muller codes.

\section{Definitions and examples}
\label{defs}

Let $\mathbb{F}_2^n$ be the space of binary vectors of length $n$ with the Hamming metric.
The {\it Hamming distance} $d(\cdot,\cdot)$ between two binary vectors is defined as the number of coordinates in which these vectors differ, while $wt(\cdot)$ denotes the \textit{weight} of a vector, i.e. the number of nonzero values it contains.
The plus sign $+$ denotes addition modulo two (componentwise in case of vectors), while the componentwise product of two binary vectors is denoted by $\ast$.

Let $X\subseteq\mathbb{F}_2^n$ be an arbitrary set and $y\in\mathbb{F}_2^n$ be an arbitrary vector. The distance from the vector $y$ to the set $X$ is defined as
\begin{equation*}
d(y,X) = \min\limits_{x\in X} d(y,x).
\end{equation*}
The {\it covering radius} of the set $X$ is defined as
\begin{equation*}
\rho(X) = \max\limits_{z\in\mathbb{F}_2^n}d(z,X).
\end{equation*}
The set $X$ with $\rho(X) = r$ is also called a {\it covering code} \cite{COH97} of radius $r$.

Consider the set
\begin{equation*}
Y=\{y\in\mathbb{F}_2^n | d(y,X)=\rho(X)\}
\end{equation*}
of all vectors at the maximal possible distance from the set $X$. This set is called the {\it metric complement} \cite{OBL16} of $X$ and is denoted by $\widehat{X}$. Vectors from the metric complement are sometimes called the {\it deep holes} of a code. If $\widehat{\widehat{X}} = X$ then the set $X$ is said to be {\it metrically regular} \cite{TOK15}.

Note that metrically regular sets always come in pairs, i.e. if $A$ is a metrically regular set, then its metric complement $\widehat{A}$ is also a metrically regular set and both of them have the same covering radius. For some simple examples of metric complements and metrically regular sets, refer to \cite{OBL16,OBL18,OBL19}.

The following trivial auxiliary lemma, established in \cite{OBL16}, will be used throughout the paper.
\begin{lemma*}
	Let $C\subseteq \mathbb{F}_2^n$ be a linear code. Then $\rho(\widehat{C}) = \rho(C)$ and a vector $x\in \mathbb{F}_2^n$ is in $\widehat{\widehat{C}}$ if and only if $x +\widehat{C} = \widehat{C}$.
\end{lemma*}

Let $\mathcal{F}^m$ be the set of all Boolean functions in $m$ variables. The Reed-Muller code of order $k$ in $m$ variables is defined as:
\begin{equation*}
    \mathcal{RM}(k,m) = \{f\in \mathcal{F}^m : \deg (f) \leqslant k\},
\end{equation*}
where $\deg(\cdot)$ denotes the degree of the \textit{algebraic normal form (ANF)}  of a function. The Reed-Muller code can be also represented as the set of \textit{value vectors} of the corresponding functions. Throughout the paper we will often switch between these two representations, sometimes ``on the fly''. In most cases, $m$ will denote the number of variables, while $n:=2^m$ will denote the dimension of the space of value vectors, which have coordinates numbered from $0$ to $2^{m}-1$. The $i$-th coordinate of a value vector is the value of the corresponding function at the binary vector of length $m$ which is the binary representation of the number $i$. Weights of functions, distances between functions and between a function and a set of functions are defined as distances between their value vectors.

Throughout the paper, vectors of length $m$ and square $m\times m$ matrices will be denoted using roman typestyle letters (e.g., $\mathrm{x}, \mathrm{A}$), while vectors of length $n$ and vectors derived from them, as well as matrices related to such vectors, will be denoted using bold letters (e.g., $\mathbf{v}, \mathbf{B}$).

Let $f$ and $g$ be two functions in $m$ variables. Let $L_{\mathrm{A}}^{\mathrm{b}}:\mathbb{F}_2^m \rightarrow \mathbb{F}_2^m$ denote the affine transformation of the variables with the matrix $\mathrm{A}$ and the vector $\mathrm{b}$:
\begin{equation*}
    (f\circ L_{\mathrm{A}}^{\mathrm{b}}) (\mathrm{x}) = f(\mathrm{Ax+b}).
\end{equation*}
Here $\circ$ denotes the operation of composition of two functions. If the vector $\mathrm{b}$ is zero, it will be omitted from the notation. Functions $f$ and $g$ are called {\it linearly equivalent} if one can be obtained from the other by applying a nonsingular linear transformation to the variables, i.e. $f = g\circ L_{\mathrm{A}}$, where $\det\mathrm{A} \neq 0$.

{\it Extended affine equivalence} is more common when classifying Boolean functions: functions $f$ and $g$ are called {\it $\mathrm{EA}$-equivalent} if there exists a nonsingular binary matrix $\mathrm{A}$, a Boolean vector $\mathrm{b}$ of length $m$ and a function $h$ of degree at most 1 such that $f = g\circ L_{\mathrm{A}}^{\mathrm{b}} + h$.

For our study we will use a variant of these two equivalence relations, which will be referred to as {\it extended linear equivalence (to the power of $k$)}. Functions $f$ and $g$ are called {\it $\mathrm{EL}^{k}$-equivalent} if there exists a nonsingular binary matrix $\mathrm{A}$ and a function $h$ of degree at most $k$ such that
\begin{equation*}
f = g\circ L_{\mathrm{A}} + h.
\end{equation*}
It is easy to see that this relation is indeed an equivalence. If two functions $f$ and $g$ are $\mathrm{EL}^{k}$-equivalent, we will denote it by $f \eleq{k} g$. We will also write $f \eeq{k} g$ if $f$ and $g$ differ by a function of degree at most $k$. Note that the last relation is a subrelation of the $\mathrm{EL}^{k}$-equivalence.

The Reed-Muller code of order $k$ in $m$ variables is usually denoted by $\mathcal{RM}(k,m)$. Since we will refer to these codes regularly, we will use $\mathcal{R}_{k,m}$ instead. The number of variables will often be omitted from the subscript if it is denoted by $m$.

\section{The Reed-Muller code $\mathcal{R}_{1,5}$} \label{RM15Section}

Let us first consider a special case --- the code $\mathcal{R}_{1,5}$. This is the set of affine functions, but in the odd number of variables, so it is not covered by the result of Tokareva concerning bent functions.

In 1972, Berlekamp and Welch presented a partition of all cosets of the code $\mathcal{R}_{1,5}$ into 48 classes with respect to the $\mathrm{EA}$-equivalence and obtained weight distributions for each class of cosets \cite{BW72}. The largest minimal weight (and therefore the covering radius of the code) among all classes is equal to $12$, and is attained on four coset classes (classes $14$, $22$, $26$ and $28$ in Table~\ref{table:RM15Cosets}). These four classes constitute the metric complement of $\mathcal{R}_{1,5}$. 

\begin{restatable}{theorem}{theoremone}
    The code $\mathcal{R}_{1,5}$ is metrically regular.
\end{restatable}

\begin{proof}
	
	\begin{table}[]
		\begin{adjustbox}{width=\columnwidth,center}
			\begin{tabular}{|l|l|l|l|l|l|}
				\hline
				No    & Representative $f$           & Added $g \in \widehat{\mathcal{R}}_{1,5}$  & $C(g)$       & Sum $h=f+g$                                              & $C(h)$ \\ \hline
				$0$   & $0$                          & ---                                        & ---          & ---                                                      & ---  \\ \hline
				$1$   & $2345$                       & $123{+}14{+}25$                            & 22           & $2345{+}123{+}14{+}25$                                   & $12$  \\ \hline
				$2$   & $2345{+}14$                  & $123{+}14{+}25$                            & 22           & $2345{+}123{+}25 \sim 2345{+}123{+}34$                   & $8$  \\ \hline
				$3$   & $2345{+}24$                  & $2345{+}123{+}24{+}35$                     & 14           & $123{+}35 \sim 123{+}14$                                 & $21$  \\ \hline
				$4$   & $2345{+}24{+}35$             & $2345{+}123{+}24{+}35$                     & 14           & $123$                                                    & $19$  \\ \hline
				$5$   & $2345{+}14{+}25$             & $123{+}14{+}25$                            & 22           & $2345{+}123$                                             & $6$  \\ \hline
				$6$   & $2345{+}123$                 & $123{+}14{+}25$                            & 22           & $2345{+}14{+}25$                                         & $5$  \\ \hline
				$7$   & $2345{+}123{+}12$            & $12{+}34$                                  & 28           & $2345{+}123{+}34$                                        & $8$  \\ \hline
				$8$   & $2345{+}123{+}34$            & $12{+}34$                                  & 28           & $2345{+}123{+}12$                                        & $7$  \\ \hline
				$9$   & $2345{+}123{+}14$            & $14{+}25$                                  & 28           & $2345{+}123{+}25 \sim 2345{+}123{+}34$                   & $8$  \\ \hline
				$10$  & $2345{+}123{+}45$            & $12{+}45$                                  & 28           & $2345{+}123{+}12$                                        & $7$  \\ \hline
				$11$  & $2345{+}123{+}12{+}34$       & $12{+}34$                                  & 28           & $2345{+}123$                                             & $6$  \\ \hline
				$12$  & $2345{+}123{+}14{+}25$       & $123{+}14{+}25$                            & 22           & $2345$                                                   & $1$  \\ \hline
				$13$  & $2345{+}123{+}12{+}45$       & $12{+}45$                                  & 28           & $2345{+}123$                                             & $6$  \\ \hline
				$14^1$& $2345{+}123{+}24{+}35$       & $2345{+}123{+}24{+}35$                     & 14           & $0$                                                      & $0$  \\ \hline
				$15$  & $2345{+}123{+}145$           & $123{+}14{+}25$                            & 22           & $2345{+}145{+}14{+}25 \sim 2345{+}123{+}12{+}34$         & $11$  \\ \hline
				$16$  & $2345{+}123{+}145{+}45$      & $123{+}145{+}45{+}24{+}35$                 & 26           & $2345{+}24{+}35$                                         & $4$  \\ \hline
				$17$  & $2345{+}123{+}145{+}24{+}45$ & $2345{+}123{+}24{+}35$                     & 14           & $145{+}35{+}45 \sim 123{+}14$                            & $21$  \\ \hline
				$18$  & $2345{+}123{+}145{+}24{+}35$ & $2345{+}123{+}24{+}35$                     & 14           & $145 \sim 123$                                           & $19$  \\ \hline
				$19$  & $123$                        & $2345{+}123{+}24{+}35$                     & 14           & $2345{+}24{+}35$                                         & $4$  \\ \hline
				$20$  & $123{+}45$                   & $2345{+}123{+}24{+}35$                     & 14           & $2345{+}24{+}35{+}45 \sim 2345{+}24{+}35$                & $4$  \\ \hline
				$21$  & $123{+}14$                   & $123{+}14{+}25$                            & 22           & $25 \sim 12$                                             & $27$  \\ \hline
				$22^2$& $123{+}14{+}25$              & $123{+}14{+}25$                            & 22           & $0$                                                      & $0$  \\ \hline
				$23$  & $123{+}145$                  & $123{+}14{+}25$                            & 22           & $145{+}14{+}25 \sim 145{+}25 \sim 123{+}14$              & $21$  \\ \hline
				$24$  & $123{+}145{+}23$             & $23{+}45$                                  & 28           & $123{+}145{+}45 \sim 123{+}145{+}23$                     & $24$  \\ \hline
				$25$  & $123{+}145{+}24$             & $123{+}15{+}24$                            & 22           & $145{+}15 \sim 123$                                      & $19$  \\ \hline
				$26^3$& $123{+}145{+}45{+}24{+}35$   & $123{+}145{+}45{+}24{+}35$                 & 26           & $0$                                                      & $0$  \\ \hline
				$27$  & $12$                         & $12{+}34$                                  & 28           & $34 \sim 12$                                             & $27$  \\ \hline
				$28^4$& $12{+}34$                    & $12{+}34$                                  & 28           & $0$                                                      & $0$  \\ \hline
			\end{tabular}
		\end{adjustbox}
		\caption{Table of even weight coset classes of $\mathcal{R}_{1,5}$. Classes marked with a superscript are the classes which constitute $\widehat{\mathcal{R}}_{1,5}$. $C(\cdot)$ denotes the No of the class the function belongs to. Functions in the table are presented in an abbreviated notation: the number $i_1 i_2 \ldots i_k$ stands for the monomial $x_{i_1} x_{i_2} \ldots x_{i_k}$. For example, the representative function for the class 14 is $x_2 x_3 x_4 x_5 + x_1 x_2 x_3 + x_2 x_4 + x_3 x_5$.} \label{table:RM15Cosets}
	\end{table}
	
	Since $\mathcal{R}_{1,5}$ is linear, it follows that $\rho(\widehat{\mathcal{R}}_{1,5}) = \rho(\mathcal{R}_{1,5}) = 12$ and $f\in \widehat{\widehat{\mathcal{R}}}_{1,5}$ if and only if $f + \widehat{\mathcal{R}}_{1,5} = \widehat{\mathcal{R}}_{1,5}$.
	Thus, in order to establish the metric regularity of $\mathcal{R}_{1,5}$, we have to prove that for every $f\notin \mathcal{R}_{1,5}$ it holds $f + \widehat{\mathcal{R}}_{1,5} \neq \widehat{\mathcal{R}}_{1,5}$.
	
	Since the covering radius of the code is even, the second metric complement of $\mathcal{R}_{1,5}$ can consist only of the cosets with codewords of even weight. There are 29 classes of such cosets, including the $\mathcal{R}_{1,5}$ code itself; they are listed in Table~\ref{table:RM15Cosets}. Classes marked with a superscript are those which constitute $\widehat{\mathcal{R}}_{1,5}$. The classification is taken from the paper by Berlekamp and Welch \cite{BW72}, but in this table some of the class representatives were modified from their original variants using simple variable swaps (for the original representatives the reader is referred to Table~\ref{table:RM1528Cosets} in the appendix of the paper).
	
	Let us show that only the $\mathcal{R}_{1,5}$ code itself is contained in the second metric complement. Let $f_{c}\notin \mathcal{R}_{1,5}$ be a function from a certain coset equivalence class $C$, and assume that the function $f_{c} + g_{c}$, where $g_{c}\in \widehat{\mathcal{R}}_{1,5}$, does not belong to any of the 4 equivalence classes from the complement $\widehat{\mathcal{R}}_{1,5}$. This implies that $f_{c} + \widehat{\mathcal{R}}_{1,5} \neq \widehat{\mathcal{R}}_{1,5}$ and thus $f_{c}$ is not in the second metric complement.
	
	Let now $f\notin \mathcal{R}_{1,5}$ be an arbitrary function from the class $C$, and let $(\mathrm{A,b},h)$ be the matrix, the vector and the affine function such that
	\begin{equation*}
	f\circ L_{\mathrm{A}}^{\mathrm{b}} + h = f_{c}.
	\end{equation*}
	Denote
	\begin{equation*}
	g_f = (g_c+h)\circ (L_{\mathrm{A}}^{\mathrm{b}})^{-1}.
	\end{equation*}
	Then the function $f + g_f$ is $\mathrm{EA}$-equivalent to $f_{c} + g_{c}$ and therefore does not belong to $\widehat{\mathcal{R}}_{1,5}$. Since $(L_{\mathrm{A}}^{\mathrm{b}})^{-1} = L_{\mathrm{A}^{-1}}^{\mathrm{A^{-1}b}}$, $g_f$ belongs to $\widehat{\mathcal{R}}_{1,5}$ and therefore $f + \widehat{\mathcal{R}}_{1,5} \neq \widehat{\mathcal{R}}_{1,5}$, which means that $f\notin \widehat{\widehat{\mathcal{R}}}_{1,5}$.
	
	Thus, if we prove that $f + g \notin \widehat{\mathcal{R}}_{1,5}$ for some $f\in C$ and some $g\in \widehat{\mathcal{R}}_{1,5}$, we will prove that no function from the equivalence class $C$ is in the second metric complement.
	
	The proof can be found in Table~\ref{table:RM15Cosets}: for a representative $f$ from each even weight coset class we find a function $g\in \widehat{\mathcal{R}}_{1,5}$ such that $f+g$ is equivalent to the representative of some class which is not in $\widehat{\mathcal{R}}_{1,5}$. Thus, the second metric complement $\widehat{\widehat{\mathcal{R}}}_{1,5}$ contains only the code $\mathcal{R}_{1,5}$ itself, proving that $\mathcal{R}_{1,5}$ is metrically regular.
	
\end{proof}

Almost all equivalences presented in the fifth column of Table 1 are variable swaps or simple additions of the form $x_i \rightarrow x_i + 1$, $x_i \rightarrow x_i + x_j$ or (for the class $20$) $x_i \rightarrow x_i + x_j + x_k$ for certain $i,j,k$.

\section{The Reed-Muller codes of orders $0$, $m$, $m-1$ and $m-2$}

The Reed-Muller codes of orders $0$, $m$ and $m-1$ coincide with the repetition code, the whole space and the even weight code respectively. It is trivial that all of them are metrically regular.

The covering radius of the Reed-Muller code of order $m-2$ is equal to $2$ \cite{COH97}. By definition, this code consists of all Boolean functions of degree at most $m-2$. Since functions of degree $m$ have odd weights, while functions of smaller degree have even weights, functions of degree $m$ are at distance $1$ from $\mathcal{R}_{m-2}$, while functions of degree $m-1$ are at distance $2$ and therefore
\begin{equation*}
\widehat{\mathcal{R}}_{m-2} = \mathcal{R}_{m-1}\setminus \mathcal{R}_{m-2}.
\end{equation*}
Since $\mathcal{R}_{m-2}$ is linear, $\rho(\widehat{\mathcal{R}}_{m-2}) = \rho(\mathcal{R}_{m-2}) = 2$ and thus functions of degree $m$ are at distance $1$ from $\widehat{\mathcal{R}}_{m-2}$. It follows that $\widehat{\widehat{\mathcal{R}}}_{m-2} = \mathcal{R}_{m-2}$ and $\mathcal{R}_{m-2}$ is metrically regular.

\section{The Reed-Muller codes of order $m-3$: Syndrome matrices} \label{SynMethSection}

McLoughlin \cite{MCL79} has proved that
\begin{equation*}
    \rho(\mathcal{R}_{m-3}) =
    \begin{cases}
        m+1, & \text{if $m$ is odd,} \\
        m+2, & \text{if $m$ is even.}
    \end{cases}
\end{equation*}
We are going to reestablish this result following the book ``Covering codes'' by Cohen et al., since our new results that follow rely on the methods and terminology described in the book. In particular, we will describe the method of obtaining the covering radius of $\mathcal{R}_{m-3}$ using syndrome matrices as it is presented in the book, with few minor adjustments. After that we will proceed to study the metric complement of $\mathcal{R}_{m-3}$. Results in Section \ref{SynMethSection} and \ref{CovRadSection}, as well as general results concerning the covering radius of $\mathcal{R}_{m-3}$, belong to Cohen et al. \cite{COH97}, while all subsequent results concerning the metric complement and the metric regularity of the code have been obtained by the author.

Let us first consider the covering radius of the punctured Reed-Muller code $\mathcal{R}_{m-3}^{\circ}$, i.e., the code without the $0$-th coordinate (which corresponds to the value of the function at zero). Let $\mathbf{H}$ denote the parity check matrix of this code. The matrix $\mathbf{H}$ coincides with the parity check matrix of the non-punctured code $\mathcal{R}_{m-3}$, but with the first all-one row and the first column removed. Since $\mathcal{R}_{m-3}$ is dual to the code $\mathcal{R}_{2}$, the rows of $\mathbf{H}$ are punctured value vectors of the functions 
\begin{equation*}
x_1,\ldots,\,x_m,\,x_1x_2,\,x_1x_3,\ldots,\,x_{m-1}x_m.
\end{equation*}

The {\it syndrome} $\mathbf{s}$ of an arbitrary vector $\mathbf{v} \in \mathbb{F}_{2}^{n-1}$ is the product $\mathbf{Hv^T}$. Let us consider the syndrome $\mathbf{s}$ as an $m\times m$ symmetric matrix $\mathbf{S}$, where the element $s_{i,j}$ of the matrix is equal to the component of the syndrome corresponding to the row $x_ix_j$ of the parity check matrix $\mathbf{H}$, while the diagonal element $s_{i,i}$ is equal to the component of the syndrome corresponding to the row $x_i$ of the matrix $\mathbf{H}$. Thus we have built a one-to-one correspondence between all cosets of $\mathcal{R}_{m-3}^{\circ}$ and all symmetric binary matrices (``syndrome matrices'').

Let $\mathbf{e_1^{\circ}},\ldots\mathbf{e_m^{\circ}} \in \mathbb{F}_{2}^{n-1}$ be the punctured value vectors of the functions $x_1,\ldots,x_m$. Notice that the row of $\mathbf{H}$ corresponding to the function $x_ix_j$ is the componentwise product $\mathbf{e_i^{\circ}}\ast\mathbf{e_j^{\circ}}$.

Consider an $m\times (n-1)$ matrix $\mathbf{B_v}$ which has $\mathbf{e_i^{\circ}}\ast \mathbf{v}$ as its $i$-th row. Then the symmetric matrix $\mathbf{S_v} = \mathbf{B_vB_v^T}$ corresponds to the syndrome $\mathbf{Hv^T}$ of the vector $\mathbf{v}$. It is easy to see that if $f$ is a function with a punctured value vector equal to $\mathbf{v}$, then the set of nonzero columns of $\mathbf{B_v}$ is precisely the support of the function $f$ (bar, possibly, the all-zero vector).
The number of nonzero columns in $\mathbf{B_v}$ is equal to the weight of the vector $\mathbf{v}$.

Given an arbitrary vector $\mathbf{v} \in \mathbb{F}_{2}^{n-1}$, its distance from the code is equal to the weight of the coset leader:
\begin{equation*}
d(\mathbf{v},\mathcal{R}_{m-3}^{\circ}) = \min\limits_{\mathbf{u}: \mathbf{Hu^T} = \mathbf{Hv^T}} wt(\mathbf{u}).
\end{equation*}
Using the established correspondences between syndromes and symmetric matrices, we can rewrite this as follows:
\begin{equation*}
d(\mathbf{v},\mathcal{R}_{m-3}^{\circ}) = \min\limits_{\mathbf{u}: \mathbf{B_u B_u^T} = \mathbf{S_v}} Col(\mathbf{B_u}),
\end{equation*}
where $Col(\mathbf{B_u})$ is the number of nonzero columns in the matrix $\mathbf{B_u}$. Let us denote the minimum on the right by $t(\mathbf{S}):=\min\limits_{\mathbf{u}: \mathbf{B_u B_u^T} = \mathbf{S}} Col(\mathbf{B_u}).$
Then
\begin{equation*}
d(\mathbf{v},\mathcal{R}_{m-3}^{\circ}) = t(\mathbf{S_v}),
\end{equation*}
and, since the correspondence between all syndromes and all symmetric matrices is one-to-one, we have
\begin{equation*}
\rho(\mathcal{R}_{m-3}^{\circ}) = \max\limits_{\mathbf{v}}d(\mathbf{v},\mathcal{R}_{m-3}^{\circ}) = \max\limits_{\mathbf{S}} t(\mathbf{S}).
\end{equation*}
Moreover, a vector $\mathbf{v}$ is in the metric complement $\widehat{\mathcal{R}}_{m-3}^{\circ}$ if and only if $t(\mathbf{S_v}) = \rho(\mathcal{R}_{m-3}^{\circ})$.

Let us call any matrix $\mathbf{B}$ such that $\mathbf{BB^T} = \mathbf{S}$ a \textit{factor} of $\mathbf{S}$. We can thus describe the value $t(\mathbf{S})$ as \textit{the minimum number of nonzero columns in a factor over all factors of $\mathbf{S}$ of the form $\mathbf{B_u}$, where $\mathbf{u}\in \mathbb{F}_2^{n-1}$}. We will call any factor achieving this minimum a \textit{minimal factor}.

Let us now expand the definition of the value $t(\mathbf{S})$.

\begin{restatable}{lemma}{lemmaone}
    Let $\mathbf{S}$ be a symmetric matrix, and let $\mathbf{B}$ be its factor (i.e. $\mathbf{B B^T} = \mathbf{S}$).
    The following operations do not change the property of $\mathbf{B}$ being a factor of $\mathbf{S}$:
    \begin{enumerate}
      \item {deleting a zero column;}
      \item {deleting two equal columns;}
      \item {swapping any two columns;}
      \item {adding an arbitrary vector $b$ to each column from some subset of columns of $\mathbf{B}$ of even size, given that all columns of this subset sum to zero.}
    \end{enumerate}
\end{restatable}
\begin{proof}
    The proof is routine and is left to the reader.
\end{proof}

Since the subsets of nonzero columns of matrices $\{\mathbf{B_u} : \mathbf{u}\in \mathbb{F}_2^{n-1}\}$ are precisely all possible subsets of nonzero columns of length $m$, Lemma 1 allows us to remove zero columns from allowed factors and ignore the possibility of duplicate columns and thus reformulate the definition of the value $t(\mathbf{S})$ in the following manner, allowing the use of arbitrarily-sized matrices:

\textit{The value $t(\mathbf{S})$ is equal to the minimum number of columns in a factor over all factors of $\mathbf{S}$. Any factor achieving this minimum is called a \textit{minimal factor} of $\mathbf{S}$.}

Moreover, any factor $\mathbf{B}$ of $\mathbf{S}$ corresponds to exactly one factor of the initial form $\mathbf{B_u}$ --- the factor with the set of nonzero columns coinciding with the set of nonzero columns of $\mathbf{B}$. Therefore, presenting any minimal factor for a symmetric matrix $\mathbf{S}$ allows us to obtain a coset leader $\mathbf{u}$ for the coset which this symmetric matrix represents.

\section{The Reed-Muller codes of order $m-3$: Covering radius} \label{CovRadSection}

In order to determine the covering radius of $\mathcal{R}_{m-3}^{\circ}$ we will now investigate the maximum possible value of $t(\mathbf{S})$. Obviously,
\begin{equation*}
t(\mathbf{S}) \geqslant \min\limits_{\mathbf{B}:\mathbf{BB^T} = \mathbf{S}}\rank(\mathbf{B}) \geqslant \rank(\mathbf{S})
\end{equation*}
for any matrix $\mathbf{S}$, and therefore
\begin{equation*}
\max\limits_{\mathbf{S}} t(\mathbf{S}) \geqslant m.
\end{equation*}

This gives us a trivial lower bound. The following proposition provides a simple upper bound:

\begin{restatable}{lemma}{lemmatwo}
	Let $\mathbf{S}$ be a symmetric matrix, and let $\mathbf{B}$ be its minimal factor. Then all proper subsets of columns of $\mathbf{B}$ are linearly independent.
\end{restatable}
\begin{proof}
	See \cite{COH97}, pp.~249--250.
	
	
	
\end{proof}

\begin{corollary}
	$t(\mathbf{S}) \leqslant m + 1$ for any symmetric $m\times m$ matrix $\mathbf{S}$.
\end{corollary}
\begin{proof}
	Assume that for some symmetric matrix $\mathbf{S}$ it holds $t(\mathbf{S}) \geqslant m+2$. This means that any minimal factor $\mathbf{B}$ of $\mathbf{S}$ has at least $m+2$ columns and therefore contains a linearly dependent proper subset of columns, which contradicts Lemma 2.
\end{proof}

This bound, combined with the previous one, shows us that the largest value of $t(\mathbf{S})$ is either $m$ or $m+1$. The following result describes the matrices with the larger value of $t(\mathbf{S})$.

\begin{lemma}
	Let $\mathbf{S}$ be a symmetric matrix. Then $t(\mathbf{S}) = m+1$ if and only if $\rank(\mathbf{S}) = m$ and $\mathbf{S}$ has an all-zero diagonal.
\end{lemma}
\begin{proof}
	
	$\Longleftarrow$
	
	Assume that the matrix $\mathbf{S}$ is nonsingular and has an all-zero diagonal, and let $\mathbf{B}$ be any of its factors. Notice that the vector consisting of all diagonal entries of the matrix $\mathbf{S}$ is the sum of all columns of $\mathbf{B}$. Therefore all columns of $\mathbf{B}$ sum to zero, which means that all its nonzero columns form a linearly dependent set of vectors. Since $\rank(\mathbf{B}) \geqslant \rank(\mathbf{S}) = m$, the matrix $\mathbf{B}$ has at least $m+1$ nonzero columns and therefore $t(\mathbf{S}) = m+1$.
	
	$\Longrightarrow$
	
	Assume that $t(\mathbf{S}) = m+1$. Let $\mathbf{B}$ be a minimal factor of $\mathbf{S}$. Note that all proper subsets of columns of $\mathbf{B}$ are linearly independent by Lemma 2, which implies that all columns of $\mathbf{B}$ sum to zero, since $\mathbf{B}$ is an $m\times (m+1)$ matrix. Since the vector consisting of diagonal elements of $\mathbf{S}$ is the sum of all columns of $\mathbf{B}$, $\mathbf{S}$ has an all-zero diagonal.
	
	Assume that $\rank(\mathbf{S}) < m$. Then there exists a subset of rows in $\mathbf{S}$ summing to $\mathbf{0}$; we denote these rows by $\mathbf{S_{i_1},S_{i_2},\ldots,S_{i_p}}$. Since $\mathbf{S_i = B_i B^T}$, this implies
	\begin{equation*}
	\mathbf{(B_{i_1}+\ldots+B_{i_p})B^T} = \mathbf{0}.
	\end{equation*}
	
	Denote $\mathbf{b} = \mathbf{B_{i_1}+\ldots+B_{i_p}}$. From the above it follows that the sum of certain columns of $\mathbf{B}$ (those corresponding to the $1$'s in the vector $\mathbf{b}$) is equal to zero.
	
	If the vector $\mathbf{b}$ is zero, then $\rank(\mathbf{B}) < m$ and it must have a linearly dependent proper subset of columns, contradiction with Lemma 2.
	
	If it is nonzero and not an all-ones vector, then we obtain a proper subset of columns of $\mathbf{B}$ which sum to $0$, contradiction with Lemma 2.
	
	Assume that $\mathbf{b}$ is an all-ones vector, and so all columns of $\mathbf{B}$ sum to zero.
	
	If $m$ is even, then the number of columns in $\mathbf{B}$ is odd and therefore $\mathbf{bb^T} = 1$, which contradicts $\mathbf{bB^T} = \mathbf{0}$.
	
	If $m$ is odd, then the number of columns in $\mathbf{B}$ is even and all rows have an even number of ones, and, by Lemma 1, we can add any column of $\mathbf{B}$ to all its columns and then remove a zero column from the resulting matrix, keeping it a factor of $\mathbf{S}$, which contradicts the minimality of $\mathbf{B}$.
	
	Thus, $\rank(\mathbf{S})$ is equal to $m$.
	
\end{proof}

Note that a matrix $\mathbf{S}$ with the properties described in the lemma (nonsingular with an all-zero diagonal) exists if and only if $m$ is even (see e.g. \cite{COH97}, p.~249). This means that
\begin{equation*}
\max\limits_{\mathbf{S}} t(\mathbf{S}) = m + 1 - \pi(m),
\end{equation*}
where $\pi(m)$ is the parity function, equal to $1$ for odd $m$ and to $0$ for even $m$.

\section{The Reed-Muller codes of order $m-3$: $m$ is even}

\subsection{\textbf{The covering radius and the metric complement of the punctured code}}

Let the number of variables $m$ be even. From previous sections we have:
\begin{equation*}
\rho(\mathcal{R}_{m-3}^{\circ}) = \max\limits_{\mathbf{S}} t(\mathbf{S}) = m+1.
\end{equation*}
A vector $\mathbf{v}\in \mathbb{F}_2^{n-1}$ is in the metric complement of $\mathcal{R}_{m-3}^{\circ}$ if and only if $t(\mathbf{S_v}) = m+1$.
The following statements will help us to characterize the syndromes of such vectors:

\begin{restatable}{lemma}{lemmathree}
    Let $\mathbf{S}$ be a symmetric $m\times m$ matrix, $m$ even. Then $t(\mathbf{S}) = m+1$ if and only if $\mathbf{S}$ has a factor of rank $m$ with $m+1$ columns which sum to zero.
\end{restatable}
\begin{proof}
	$\Longrightarrow$
	
	Assume that $t(\mathbf{S}) = m+1$ and let $\mathbf{B}$ be an arbitrary minimal factor of $\mathbf{S}$. By Proposition 1, $\rank(\mathbf{S}) = m$ and $\mathbf{S}$ has an all-zero diagonal. Therefore, $\rank(\mathbf{B}) = m$ and all its columns sum to zero.
	
	$\Longleftarrow$
	
	Let $\mathbf{B}$ be a factor of $\mathbf{S}$ of rank $m$ with $m+1$ columns which sum to $\mathbf{0}$.
	
	Assume that $t(\mathbf{S}) = k \leqslant m$ and let $\mathbf{D}$ be an arbitrary minimal factor of $\mathbf{S}$ with $k$ columns. Since the sum of all columns of a factor is the vector consisting of the diagonal elements of $\mathbf{S}$, the sum of all columns of $\mathbf{D}$ is also equal to zero. This implies that $\rank(\mathbf{D}) < m$, and therefore $\rank(\mathbf{S}) < m$.
	
	It is easy to see that each proper subset of columns of $\mathbf{B}$ is linearly independent. Notice that the existence of a factor with this property is shown to contradict with the assumption ``$\rank(\mathbf{S}) < m$'' in the proof of Proposition 1, and in the case when $m$ is even the proof does not rely on the minimality of $\mathbf{B}$.
	
	Thus, $t(\mathbf{S}) = m+1$ and $\mathbf{B}$ is a minimal factor of $\mathbf{S}$.
	
\end{proof}

It is easy to see that Lemma 4 describes all minimal factors of all matrices $\mathbf{S}$ satisfying $t(\mathbf{S}) = m+1$. Let us construct the following set:
\begin{multline*}
  \mathrm{U} = \{\mathbf{u}\in\mathbb{F}_2^{n-1} : \mathbf{B_u} \text{ has $m+1$ nonzero columns, $m$ of which are} \\
  \text{linearly independent and all of them sum to zero}\}.
\end{multline*}
Trivially, the set of matrices $\{\mathbf{B_u} : \mathbf{u}\in \mathrm{U}\}$ (up to columns permutations and zero columns removal) includes exactly all minimal factors described in Lemma 4. Therefore, if $t(\mathbf{S}) = m+1$ for some matrix $\mathbf{S}$, then there exists a vector $\mathbf{u}\in \mathrm{U}$ such that $\mathbf{S} = \mathbf{B_u B_u^T}$. Conversely, for any $\mathbf{u}\in \mathrm{U}$ it holds $t(\mathbf{B_u B_u^T}) = m+1$. Thus, the vectors from the set $\mathrm{U}$ cover all cosets contained in the metric complement of $\mathcal{R}^{\circ}_{m-3}$:
\begin{equation*}
\widehat{\mathcal{R}}^{\circ}_{m-3} = \bigcup\limits_{\mathbf{u}\in \mathrm{U}} \left(\mathbf{u} + \mathcal{R}_{m-3}^{\circ}\right).
\end{equation*}

\subsection{\textbf{The covering radius and the metric complement of the non-punctured code}}

We have obtained the covering radius and described the metric complement of the punctured code. Let us return to the regular, non-punctured Reed-Muller code $\mathcal{R}_{m-3}$. Since it is obtained from the punctured code by adding a parity check bit at the $0$-th coordinate, the following result will be of use:
\begin{restatable}{lemma}{lemmafour}
    Let $C$ be a code with the covering radius $r$ and the metric complement $\widehat{C}$. Let $C_{\pi}$ be the code obtained from $C$ by adding a parity check bit to all codewords of $C$ (in the front). Then $\rho(C_{\pi}) = r+1$ and $\widehat{C}_{\pi}$ is obtained from $\widehat{C}$ by
    \begin{enumerate}
      \item{adding a parity check bit to all vectors in case if $r$ is odd or}
      \item{adding an inversed parity check bit to all vectors in case if $r$ is even.}
    \end{enumerate}
\end{restatable}
\begin{proof}
	Obviously, $\rho(C_{\pi}) \leqslant r+1$. 
	
	Let us prove (2). Assume that $r$ is even. Denote
	\begin{equation*}
	C_i = \{c\in C: wt(c) \bmod 2 = i \},\,\,\widehat{C}_i = \{c\in \widehat{C}: wt(c) \bmod 2 = i \},\,\,i=0,1.
	\end{equation*}
	Since $r$ is even, vectors from $\widehat{C}_0$ are at distance $r$ from $C_0$ and at a larger distance from $C_1$. Similarly, vectors from $\widehat{C}_1$ are at distance $r$ from $C_1$ and at a larger distance from $C_0$.
	
	Let $c' = (\epsilon, c)$, where $c\notin \widehat{C}$ and $\epsilon \in \{0,1\}$. Then $d(c', C_{\pi}) \leqslant d(c, C) + 1 \leqslant r$.
	
	Let $c \in \widehat{C}_1$. Then $d((1,c), C_{\pi}) = \min(d(c, C_1), d(c, C_0)+1) = r$, while $d((0,c), C_{\pi}) = \min(d(c, C_1)+1, d(c, C_0)) > r$.
	
	Let $c \in \widehat{C}_0$. Similarly, $d((1,c), C_{\pi}) = \min(d(c, C_1), d(c, C_0)+1) > r$, while $d((0,c), C_{\pi}) = \min(d(c, C_1)+1, d(c, C_0)) = r$.
	
	Therefore, vectors $\{(1,c) | c \in \widehat{C}_0\} \cup \{(0,c) | c \in \widehat{C}_1\}$ are the only ones at a distance larger than $r$ from $C_{\pi}$, and this distance can be only equal to $r+1$. The claim (2) of the lemma is proved.
	
	The proof of the case (1) is completely similar to the above, but with some sets switched around.
	
\end{proof}
Using this lemma we find that the covering radius of the non-punctured Reed-Muller code $\mathcal{R}_{m-3}$ is equal to $m+2$ and its metric complement can be described as follows:
\begin{equation*}
\widehat{\mathcal{R}}_{m-3} = \bigcup\limits_{\mathbf{u}\in \mathrm{U}} \left((\pi(\mathbf{u}),\mathbf{u}) + \mathcal{R}_{m-3}\right).
\end{equation*}

Let $f_{\mathbf{v}}$ denote the function with the value vector $\mathbf{v}\in\mathbb{F}_2^n$ (non-punctured). Recall that the set of nonzero columns of the matrix $\mathbf{B_{v^{\circ}}}$ coincides with the support of the function $f_{\mathbf{v}}$, bar, possibly, the zero vector. Since all vectors in $\mathrm{U}$ have odd weights and added parity check bit corresponds to the value of the function at the all-zero vector, we can describe the metric complement of $\mathcal{R}_{m-3}$ in terms of functions instead of their value vectors as follows:
\begin{equation*}
\widehat{\mathcal{R}}_{m-3} = \bigcup\limits_{g\in G} \left(g + \mathcal{R}_{m-3}\right),
\end{equation*}
where
\begin{multline*}G = \{f_{(1,\mathbf{u)}} : \mathbf{u}\in \mathrm{U}\} = \{g : \supp(g) = \{\mathrm{0},\mathrm{x}_1,\mathrm{x}_2\ldots,\mathrm{x}_m,\mathrm{x}_1+\ldots+\mathrm{x}_m\}, \\
\{\mathrm{x}_1,\ldots,\mathrm{x}_m\} \text{ are linearly independent}\}.
\end{multline*}

All functions in $G$ form an equivalence class with respect to the linear equivalence. Recall that two functions $f$ and $g$ are called {\it $\mathrm{EL}^{k}$-equivalent} if there exists a nonsingular binary matrix $\mathrm{A}$ and a function $h$ of degree at most $k$ such that $g = f\circ L_{\mathrm{A}} + h$. It is now easy to see that a function $g$ is in $\widehat{\mathcal{R}}_{m-3}$ if and only if it is $\mathrm{EL}^{m-3}$-equivalent to some function from $G$. Since all functions in the metric complement are equivalent, we can pick any function from it as the reference for equivalence (and we will change this reference when it is convenient). We will call the $\mathrm{EL}^{m-3}$-equivalence just ``equivalence'' for brevity from now on.

Let us give an explicit (algebraic normal form) description of a certain function from $G$.
Denote by $g^*$ the function with the support $\{\mathrm{0},\mathrm{e}_1,\mathrm{e}_2,\ldots,\mathrm{e}_m,\mathrm{1}\}$, where $\mathrm{e}_i\in \mathbb{F}_2^m$ is the vector with $1$ only in the $i$-th coordinate. Clearly, $g^* \in G$ and it is straightforward to construct the algebraic normal form of this function: it is the sum of all monomials containing an even number of variables, excluding the monomial with all variables included:
\begin{equation*}
g^*(\mathrm{x}) = 1 + \sum\limits_{k = 1}^{\frac{m}{2}-1} \sum\limits_{1\leqslant i_1<\ldots < i_{2k}\leqslant m} x_{i_1} x_{i_2}\ldots x_{i_{2k}}.
\end{equation*}
This function is equivalent to the sum of all monomials containing $m-2$ variables, so let us use this last function as $g^*$ moving forward.
Let $\overline{x_i}$ denote the product of all $m$ variables except $x_i$, and let $\overline{x_i x_j}$ denote the product of all $m$ variables except $x_i$ and $x_j$. Using these conventions, we can write this new representative function as follows:
\begin{equation*}
g^*(\mathrm{x}) := \sum\limits_{1\leqslant i < j\leqslant m} \overline{x_i x_j}.
\end{equation*}

\subsection{\textbf{Metric regularity}}

We have established that
\begin{equation*}
\widehat{\mathcal{R}}_{m-3} = \{g : g\eleq{m-3} g'\},
\end{equation*}
where $g'$ is an arbitrary function from the class $G$ (or from $\widehat{\mathcal{R}}_{m-3}$), and have constructed a certain representative of this equivalence class --- $g^*$.

Since the code $\mathcal{R}_{m-3}$ is linear, $\rho(\widehat{\mathcal{R}}_{m-3}) = \rho(\mathcal{R}_{m-3}) = m+2$ and a function $f$ is in $\widehat{\widehat{\mathcal{R}}}_{m-3}$ if and only if $f + \widehat{\mathcal{R}}_{m-3} = \widehat{\mathcal{R}}_{m-3}$.
Let us prove the metric regularity of $\mathcal{R}_{m-3}$ by proving that no functions other that those contained in $\mathcal{R}_{m-3}$ preserve the metric complement under addition.

Let $f\notin \mathcal{R}_{m-3}$ be an arbitrary function. Since $\widehat{\mathcal{R}}_{m-3}$ is an $\mathrm{EL}^{m-3}$-equivalence class, in order to show that $f + \widehat{\mathcal{R}}_{m-3} \neq \widehat{\mathcal{R}}_{m-3}$ it is enough to show there exists a function $f'$ such that $f' \eleq{m-3} f$ and $f' + \widehat{\mathcal{R}}_{m-3} \neq \widehat{\mathcal{R}}_{m-3}$.

\textbf{Case 1.} Let $f\notin \mathcal{R}_{m-3}$ be a function of degree greater than $m-2$. Since $\mathrm{EL}^{m-3}$-equivalence preserves the degree for functions of degree higher than $m-3$, any $g\in\widehat{\mathcal{R}}_{m-3}$ has degree $m-2$ (like $g^*$), while $f+g$ has a higher degree and therefore cannot be equivalent to any of the functions from $\widehat{\mathcal{R}}_{m-3}$. Thus, functions of degree greater than $m-2$ do not preserve any function from the metric complement and therefore cannot be in $\widehat{\widehat{\mathcal{R}}}_{m-3}$.

\textbf{Case 2.} Let $f\notin \mathcal{R}_{m-3}$ be a function of degree $m-2$. We can uniquely present it as follows:
\begin{equation*}
f(\mathrm{x}) = \sum\limits_{(i,j)\in I} \overline{x_i x_j} + h(\mathrm{x}),
\end{equation*}
where $\deg(h) < m-2$. Denote by $\tilde{f}$ the following quadratic function:
\begin{equation*}
\tilde{f}(\mathrm{x}) := \sum\limits_{(i,j)\in I} x_i x_j.
\end{equation*}
We will call $\tilde{f}$ the \textit{quadratic dual} of $f$.

The following result would be of use when handling this case:
\begin{restatable}{lemma}{lemmafive}
    Let $f$ and $g$ be two functions of degree $m-2$. Then $f\eleq{m-3} g$ if and only if their quadratic duals are $\mathrm{EL}^1$-equivalent ($\mathrm{EA}$-equivalent).
\end{restatable}
\begin{proof}
	Since $\mathrm{EL}^{m-3}$-equivalence allows us to add functions of degree up to $m-3$, we will assume that both $f$ and $g$ contain only monomials of degree $m-2$. In what follows we will discard monomials of degree less than $m-2$ when talking about $\mathrm{EL}^{m-3}$-equivalence, and we will discard monomials of degree less than $2$ when talking about $\mathrm{EL}^1$-equivalence.
	
	Let $f(\mathrm{x}) = \sum\limits_{(i,j)\in I} \overline{x_i x_j}$ be the ANF of $f$. Let us perform the following simple nonsingular linear transformation of variables $L_{ij}$:
	\begin{equation*}
	L_{ij}:
	\begin{cases}
	x_i\leftarrow x_i+x_j, &  \\
	x_k\leftarrow x_k & \forall k\neq i.
	\end{cases}
	\end{equation*}
	The function $f$ changes under this transformation (disregarding monomials of degree less than $m-2$) in the following manner:
	\begin{equation*}
	L_{ij}:
	\begin{cases}
	\overline{x_i x_k}\leftarrow \overline{x_i x_k} & \forall k\neq i, \\
	\overline{x_j x_k}\leftarrow \overline{x_j x_k}+\overline{x_i x_k} & \forall k\neq i,j, \\
	\overline{x_k x_l}\leftarrow \overline{x_k x_l} & \forall k,l\neq i,j.
	\end{cases}
	\end{equation*}
	
	Let $f_1$ denote the function obtained after this transformation. Then it is easy to see that the dual function $\tilde{f}_1$ is obtained from the dual function $\tilde{f}$ (disregarding monomials of degree less than $2$ since we consider $\mathrm{EL}^1$-equivalence) by the following linear transformation:
	\begin{equation*}
	L_{ji}:
	\begin{cases}
	x_j\leftarrow x_j+x_i, &  \\
	x_k\leftarrow x_k & \forall k\neq j.
	\end{cases}
	\end{equation*}
	which is simply the transposed transformation.
	
	Assume now that $g$ is obtained from $f$ using some linear transformation $L$. Trivially, $L$ can be decomposed into a sequence of simple transformations:
	\begin{equation*}
	L = L_{i_1 j_1}\circ L_{i_2 j_2}\circ\ldots\circ L_{i_s j_s}.
	\end{equation*}
	From the above we can see that the dual function $\tilde{g}$ is obtained from $\tilde{f}$ using the following transformation $\tilde{L}$:
	\begin{equation*}
	\tilde{L} = L_{j_1 i_1}\circ L_{j_2 i_2}\circ\ldots\circ L_{j_s i_s}
	\end{equation*}
	which is a sequence of transposed simple transformations.
	
	Thus we have established that, if $f\eleq{m-3} g$, then $\tilde{f}\eleq{1} \tilde{g}$. The reverse can be shown using similar argumentation.
	
\end{proof}

It is known that any quadratic Boolean function is $\mathrm{EA}$-equivalent to the function of the form $x_1 x_2 + x_3 x_4 + \ldots + x_{2k-1} x_{2k}$ for some $k\leqslant \frac{m}{2}$, and any two functions of this form with different number of variables are not $\mathrm{EA}$-equivalent one to the other. Using this result and Lemma 6 we conclude that $f$ is equivalent to the function $p_k$ for some $k$ $(0 < k \leqslant \frac{m}{2})$, where
\begin{equation*}
p_k(\mathrm{x}) = \overline{x_{1}x_{2}} + \overline{x_{3}x_{4}} + \ldots + \overline{x_{2k-1}x_{2k}} = \sum\limits_{i=1}^{k} \overline{x_{2i-1}x_{2i}}.
\end{equation*}

Trivially, $g^*$ is equivalent to $p_{\frac{m}{2}}$. Then $p_k+p_{\frac{m}{2}}$ is equivalent to $p_{\frac{m}{2}-k}$, which is (by Lemma 6) not equivalent to $p_{\frac{m}{2}}$ and therefore not equivalent to $g^*$. This means that $f + \widehat{\mathcal{R}}_{m-3} \neq \widehat{\mathcal{R}}_{m-3}$ and therefore $f\notin \widehat{\widehat{\mathcal{R}}}_{m-3}$.

Since all functions which are not in $\mathcal{R}_{m-3}$ have degree $m-2$ or higher, we have just shown that none of them are in the second metric complement, and therefore $\mathcal{R}_{m-3}$ is metrically regular when $m$ is even.

\section{The Reed-Muller codes of order $m-3$: $m$ is odd}

\subsection{\textbf{The covering radius and the metric complement of the punctured code}}

Let the number of variables $m$ be odd. Many arguments for this case are similar or identical to the ones for the previous case, however, the proof a bit more complicated. From Section 6 we have:
\begin{equation*}
\rho(\mathcal{R}_{m-3}^{\circ}) = \max\limits_{\mathbf{S}} t(\mathbf{S}) = m,
\end{equation*}
and a vector $\mathbf{v}$ is in the metric complement of $\mathcal{R}_{m-3}^{\circ}$ if and only if $t(\mathbf{S_v}) = m$. The following lemma will help to characterize matrices achieving this maximum:
\begin{restatable}{lemma}{lemmasix}
    Let $\mathbf{S}$ be a symmetric $m\times m$ matrix, where $m$ is odd. Then $t(\mathbf{S}) = m$ if and only if $\mathbf{S}$ has an $m\times m$ factor which is either nonsingular, or has rank $m-1$ and all columns summing to zero.

\end{restatable}
\begin{proof}
	
	$\Longrightarrow$
	
	Assume that $t(\mathbf{S}) = m$ and let $\mathbf{B}$ be a minimal factor of $\mathbf{S}$ with $m$ columns. If the rank of $\mathbf{B}$ is smaller than $m-1$, then $\mathbf{B}$ has a proper subset of columns summing to zero, contradicting the minimality of $\mathbf{B}$, so the rank of the factor must be at least $m-1$. If the rank is $m$, the proof is finished.
	
	Assume that $\rank(\mathbf{B}) = m-1$. Then some subset of columns of $\mathbf{B}$ must sum to zero. Since $\mathbf{B}$ is minimal, it cannot be a proper subset by Lemma 2, therefore all columns of $\mathbf{B}$ must sum to zero.
	
	$\Longleftarrow$
	
	Clearly, $t(\mathbf{S}) \geqslant \rank(\mathbf{S})$, so if $\mathbf{S} = \mathbf{BB^T}$ for some nonsingular $m\times m$ matrix $\mathbf{B}$, then the proof is finished.
	
	Let $\mathbf{S} = \mathbf{BB^T}$ for some $\mathbf{B}$ of rank $m-1$ with all columns summing to zero.
	
	Assume that $t(\mathbf{S}) = k \leqslant m-1$ and let $\mathbf{D}$ be a minimal factor of $\mathbf{S}$.
	Since the sum of all columns of any factor is the vector composed of the diagonal elements of $\mathbf{S}$,  the sum of all columns of $\mathbf{D}$ is also zero.
	
	Assume that $k = m-1$. Then $\mathbf{D}$ has an even number of columns, and each row has an even number of ones, so we can add an arbitrary vector to all columns of $\mathbf{D}$ while keeping it a factor of $\mathbf{S}$ using Lemma 1. Let us add the first column of $\mathbf{D}$ to all its columns. Now the first column of $\mathbf{D}$ is zero and we can remove it by Lemma 1. We have now obtained a factor of $\mathbf{S}$ with fewer columns than $\mathbf{D}$ has, which contradicts the minimality of $\mathbf{D}$.
	
	Therefore, $k$ can be at most $m-2$. Since all its columns sum to zero, $\mathbf{D}$ is not a full-rank matrix. Hence $\rank(\mathbf{D})$ is at most $m-3$, which means that $\rank(\mathbf{S})$ is at most $m-3$ as well.
	
	Since $\mathbf{S} = \mathbf{BB^T}$, by Sylvester's inequality we obtain $\rank(\mathbf{S}) \geqslant \rank(\mathbf{B}) + \rank(\mathbf{B^T}) - m = m-2$. But we have just established that $\rank(\mathbf{S}) \leqslant m-3$, contradiction. 
	
	Thus, $t(\mathbf{S})$ has to be greater than $m-1$ and is equal to $m$.
	
\end{proof}

Lemma 7 describes all minimal factors of all matrices $\mathbf{S}$ satisfying $t(\mathbf{S}) = m$. Let us put
\begin{multline*}
  \mathrm{U}_1 = \{\mathbf{u} : \mathbf{B_u} \text{ has $m$ nonzero columns which are linearly independent}\}
\end{multline*}
and
\begin{multline*}
  \mathrm{U}_2 = \{\mathbf{u} : \mathbf{B_u} \text{ has $m$ nonzero columns, $m-1$ of which are} \\
  \text{linearly independent and the sum of all columns is equal to zero}\}.
\end{multline*}

Denote $\mathrm{U} = \mathrm{U}_1 \cup \mathrm{U}_2$. It is easy to see that the set of matrices $\{\mathbf{B_u} : \mathbf{u}\in \mathrm{U}\}$ (up to columns permutations and zero columns removal) includes exactly all minimal factors described in Lemma 7. Thus, if $t(\mathbf{S}) = m$ for some matrix $\mathbf{S}$, then there exists a vector $\mathbf{u}\in \mathrm{U}$ such that $\mathbf{S} = \mathbf{B_u B_u^T}$. Conversely, for any $\mathbf{u}\in \mathrm{U}$ it holds $t(\mathbf{B_u B_u^T}) = m$. Therefore, the vectors from the set $\mathrm{U}$ cover all cosets contained in the metric complement of $\mathcal{R}^{\circ}_{m-3}$:

\begin{equation*}
\widehat{\mathcal{R}}^{\circ}_{m-3} = \bigcup\limits_{\mathbf{u}\in \mathrm{U}} (\mathbf{u} + \mathcal{R}_{m-3}^{\circ}).
\end{equation*}

\subsection{\textbf{The covering radius and the metric complement of the non-punctured code}}

Let us return to the regular, non-punctured Reed-Muller code $\mathcal{R}_{m-3}$. As with the case when $m$ is even, since the code is obtained from the punctured one by adding a parity check bit, using Lemma 5 we conclude that the covering radius of $\mathcal{R}_{m-3}$ is equal to $m+1$, and its metric complement is
\begin{equation*}
\widehat{\mathcal{R}}_{m-3} = \bigcup\limits_{\mathbf{u}\in \mathrm{U}} ((\pi(\mathbf{u}),\mathbf{u}) + \mathcal{R}_{m-3}).
\end{equation*}

Recall once again that for any $\mathbf{v}\in\mathbb{F}_2^n$, the set of nonzero columns of $\mathbf{B_{v^{\circ}}}$ coincides with the support of the function $f_{\mathbf{v}}$, bar, possibly, the zero vector. Since all vectors in $\mathrm{U}$ have odd weight and added parity check bit corresponds to the value of the function at the all-zero vector, we can rewrite the metric complement of $\mathcal{R}_{m-3}$ in terms of functions instead of their value vectors:

\begin{equation*}
\widehat{\mathcal{R}}_{m-3} = \bigcup\limits_{g\in G_1 \cup G_2} g + \mathcal{R}_{m-3},
\end{equation*}
where
\begin{multline*}
G_1 = \{f_{(1,\mathbf{u)}} : \mathbf{u}\in \mathrm{U}_1\} = \\
= \{g : \supp(f) = \{\mathrm{0},\mathrm{x}_1,\mathrm{x}_2\ldots,\mathrm{x}_m\}, \{\mathrm{x}_1,\ldots,\mathrm{x}_m\} \text{ are linearly independent}\},
\end{multline*}
and
\begin{multline*}
G_2 = \{f_{(1,\mathbf{u)}} : \mathbf{u}\in \mathrm{U}_2\} = \\
= \{g : \supp(g) = \{\mathrm{0},\mathrm{x}_1,\mathrm{x}_2\ldots,\mathrm{x}_{m-1},\mathrm{x}_1+\ldots+\mathrm{x}_{m-1}\}, \\
\{\mathrm{x}_1,\ldots,\mathrm{x}_{m-1}\} \text{ are linearly independent}\}.
\end{multline*}

It is easy to see that all functions in $G_1$ form an equivalence class with respect to the linear equivalence, so do functions in $G_2$. Let us pick two arbitrary functions $g_1\in G_1$, $g_2\in G_2$ from these two classes.
Then it follows from the definition of the $\mathrm{EL}^{k}$-equivalence that a function $g$ is in $\widehat{\mathcal{R}}_{m-3}$ if and only if $g \eleq{m-3} g_1$ or $g \eleq{m-3} g_2$.
In fact, we can pick any function from the $\mathrm{EL}^{m-3}$-equivalence class of $G_1$ and from the $\mathrm{EL}^{m-3}$-equivalence class of $G_2$ respectively as our references of equivalence.

Let us give an explicit (algebraic normal form) description of a certain function from $G_1$. Denote by $g_1^*$ the function with the support $\{\mathrm{0},\mathrm{e}_1,\mathrm{e}_2,\ldots,\mathrm{e}_{m-1},\mathrm{1}\}$. After a bit of calculation one can explicitly describe its ANF:
\begin{equation*}
g_1^*(\mathrm{x}) = \overline{x_m} + (1+x_m)\left( 1+\sum\limits_{k = 1}^{\frac{m-3}{2}} \sum\limits_{1\leqslant i_1<\ldots < i_{2k}\leqslant m-1} x_{i_1} x_{i_2}\ldots x_{i_{2k}} \right).
\end{equation*}
This function has degree $m-1$ and, omitting all terms of degree less than $m-2$, it is trivially $\mathrm{EL}^{m-3}$-equivalent to the following function which we will use as $g_1^*$ from now on:
\begin{equation}
g_1^* := \overline{x_m} + x_m g^{\star},
\end{equation}
where $g^{\star}$, defined by
\begin{equation*}
g^{\star}(x_1,x_2,\ldots,x_{m-1}) = \left(\sum\limits_{1\leqslant i < j\leqslant m-1} \overline{x_i x_j}\right),
\end{equation*}
is a function of the first $m-1$ variables. Moving on we will denote the $(m-1)$-tuple of the first $m-1$ variables as $\bar{\mathrm{x}}$. We will also denote affine transformations of the first $m-1$ variables as $\bar{L}_{\mathrm{A}}^{\mathrm{b}}$ (with the matrix and the vector of corresponding sizes).

Let us now give an explicit description of a certain function from $G_2$. Denote by $g_2^*$ the function with the support $\{\mathrm{0},\mathrm{e}_1,\mathrm{e}_2,\ldots,\mathrm{e}_{m-1},\sum\limits_{i=1}^{m-1} \mathrm{e}_i\}$. After a bit of calculation one can explicitly describe its ANF:
\begin{equation*}
g_2^*(\mathrm{x}) = (1+x_m) \left( 1+\sum\limits_{k = 1}^{\frac{m-3}{2}} \sum\limits_{1\leqslant i_1<\ldots < i_{2k}\leqslant m-1} x_{i_1} x_{i_2}\ldots x_{i_{2k}} \right).
\end{equation*}
This function has degree $m-1$ and is trivially $\mathrm{EL}^{m-3}$-equivalent to the function $x_m g^{\star}$, which we will use as $g_2^*$ from now on:
\begin{equation}
g_2^* := x_m g^{\star}
\end{equation}

Note that $g_1^* = \overline{x_m} + g_2^*$.

Before we proceed to establish the metric regularity of $\mathcal{R}_{m-3}$, we will build some alternative representatives of the equivalence classes of $G_1$ and $G_2$. The following lemma will be helpful:
\begin{restatable}{lemma}{lemmaseven}
    Let $f$ be a function such that $f \eeq{m-2} \overline{x_m}$. Let $\mathrm{A}$ be a nonsingular $m\times m$ matrix. Then $f\circ L_{\mathrm{A}} \eeq{m-2} \overline{x_m}$ if and only if the matrix $\mathrm{A}$ has the following form:
    \begin{equation*}
    \mathrm{A} = \begin{pmatrix}
    \bar{\mathrm{A}} & \mathrm{0}^{m-1} \\
    \mathrm{w} & 1
    \end{pmatrix},
    \end{equation*}
    where $\mathrm{0}^{m-1}$ is an all-zero column of length $m-1$, $\bar{\mathrm{A}}$ is an arbitrary nonsingular $(m-1)\times (m-1)$ matrix and $\mathrm{w}$ is an arbitrary row of length $m-1$.
\end{restatable}
\begin{proof}
	
	$\Longleftarrow$
	
	Trivially, such transformation of the first $m-1$ variables keeps the monomial $\overline{x_m}$ in $f$ the only monomial of degree $(m-1)$, and the linear transformation cannot increase the degree of any of the other monomials.
	
	$\Longrightarrow$
	
	Assume that $f\circ L_{\mathrm{A}} \eeq{m-2} \overline{x_m}$. This means that the change of variables keeps the monomial $\overline{x_m}$ intact and does not produce any other monomials of degree $m-1$. Clearly, the action of this change on monomials of degree $m-2$ and smaller is irrelevant, so let us inspect the action on $\overline{x_m}$.
	
	It is easy to see that the coefficient of the monomial $\overline{x_i}$ in the resulting function, obtained after applying transformation $L_{\mathrm{A}}$ to the variables, is precisely the value of the $(m-1)\times (m-1)$ minor, obtained from the matrix $\mathrm{A}$ by removing the $m$-th row and the $i$-th column. So we need the matrix $\mathrm{A}$ to have all such minors be equal to zero, except for the last one, obtained by removing the last column.
	
	Let $\bar{\mathrm{A}}^1,\bar{\mathrm{A}}^2,\ldots,\bar{\mathrm{A}}^m$ denote the columns of the matrix $\mathrm{A}$ with the last coordinate removed. 
	Then the condition on the minors described above can be reformulated as follows: sets of columns $\{\bar{\mathrm{A}}^1,\ldots,\bar{\mathrm{A}}^{i-1},\bar{\mathrm{A}}^{i+1},\ldots,\bar{\mathrm{A}}^m\}$ are linearly dependent for all $i\neq m$, while the set of the first $m-1$ columns is linearly independent. This implies that the following set of equations holds:
	\begin{equation*}
	\begin{cases}
	\bar{\mathrm{A}}^m + \sum\limits_{j\leqslant m-1} b_{1,j}\bar{\mathrm{A}}^j = \mathrm{0} \\
	\bar{\mathrm{A}}^m + \sum\limits_{j\leqslant m-1} b_{2,j}\bar{\mathrm{A}}^j = \mathrm{0} \\
	\ldots \\
	\bar{\mathrm{A}}^m + \sum\limits_{j\leqslant m-1} b_{m-1,j}\bar{\mathrm{A}}^j = \mathrm{0} \\
	\end{cases}
	\end{equation*}
	where $\mathrm{B} = (b_{i,j})$ --- the coefficients matrix --- is an $(m-1)\times (m-1)$ matrix with $b_{i,i} = 0$ for all $i$.
	
	If we denote the rows of the matrix $\mathrm{B}$ by $\mathrm{B}_i$, and denote by $\bar{\mathrm{A}}$ the $(m-1)\times (m-1)$ matrix composed of the first $m-1$ columns $\bar{\mathrm{A}}^{1},\ldots,\bar{\mathrm{A}}^{m-1}$, we can rewrite this in the following manner:
	\begin{equation*}
	\begin{cases}
	\bar{\mathrm{A}} \cdot \mathrm{B}_1^T = \bar{\mathrm{A}}^m \\
	\bar{\mathrm{A}} \cdot \mathrm{B}_2^T = \bar{\mathrm{A}}^m \\
	\ldots \\
	\bar{\mathrm{A}} \cdot \mathrm{B}_{m-1}^T = \bar{\mathrm{A}}^m \\
	\end{cases}
	\end{equation*}
	Since $\bar{\mathrm{A}}$ is nonsingular, the solution to each equation (which is a system of equations on $b_{i,j}$'s for $i$-th row) is unique and hence $\mathrm{B}_1=\mathrm{B}_2=\ldots = \mathrm{B}_{m-1}$. Since $b_{i,i} = 0$, the matrix $\mathrm{B}$ is a zero matrix, which means that $\bar{\mathrm{A}}^m = 0$. This implies that the last column of the matrix $\mathrm{A}$ can have $1$ only in the last coordinate, and since $\mathrm{A}$ is nonsingular, this has to be the case. Thus, $\mathrm{A}$ is of the form stated in the lemma.
	
\end{proof}

This lemma shows us that all linear transformations of the described form, and only such transformations among all linear, transform functions of the form $\overline{x_m} + h$ with $\deg(h) \leqslant m-2$ into functions of the same form, preserving $\overline{x_m}$ as the only monomial of degree $m-1$. Let us look closer at how such transformations act on monomials of degree $m-2$ in such functions:

\begin{corollary}
Let $f$ be a function of degree $m-1$ such that
\begin{equation*}
f = \overline{x_m} + x_m f_1 + f_2,
\end{equation*}
where $f_1, f_2$ do not depend on $x_m$ and $\deg(f_1) \leqslant m-3$, $\deg(f_2) \leqslant m-2$.
Let $\mathrm{A}$ be a matrix satisfying the conditions of Lemma 8. Then
\begin{equation*}
f\circ L_{\mathrm{A}} \eeq{m-3} \overline{x_m} + x_m (f_1\circ \bar{L}_{\bar{\mathrm{A}}}) + f_3,
\end{equation*}
where $f_3$ is some function of degree at most $m-2$ which does not depend on the variable $x_m$.
\end{corollary}
\begin{proof}
	Straighforward from the proof of Lemma 8.
\end{proof}

Let us now build alternative representatives for the metric complement of $\mathcal{R}_{m-3}$. Since $\bar{\mathrm{A}}$ in Lemma 8 can be any nonsingular matrix, choosing $\bar{\mathrm{A}}$ so that $g^{\star}\circ \bar{L}_{\bar{\mathrm{A}}} \eeq{m-3} p_{\frac{m-1}{2}}$, (this is possible by Lemma 6) and filling the vector $\mathrm{w}$ with zeroes, we obtain a matrix $\mathrm{A}$ such that
\begin{equation}
g_{1}^{**} := g_1^*\circ L_{\mathrm{A}} \eeq{m-3} \overline{x_m} + x_m (g^{\star}\circ \bar{L}_{\bar{\mathrm{A}}}) + h_1 \eeq{m-3} \overline{x_m} + x_m p_{\frac{m-1}{2}} + h_1.
\end{equation}
Here $p_{\frac{m-1}{2}}, h_1$ do not depend on $x_m$ and $h_1$ has degree at most $m-2$. Additionally,
\begin{equation}
g_{2}^{**} := g_2^*\circ L_{\mathrm{A}} \eeq{m-3} x_m (g^{\star}\circ \bar{L}_{\bar{\mathrm{A}}}) \eeq{m-3} x_m p_{\frac{m-1}{2}}.
\end{equation}
We will use these equivalent functions $g_{1}^{**}$ and $g_{2}^{**}$ as class representatives in some cases.

\subsection{\textbf{Metric regularity}}

We have established that
\begin{equation*}
\widehat{\mathcal{R}}_{m-3} = \{g : g\eleq{m-3} g_1\} \cup \{g : g\eleq{m-3} g_2\},
\end{equation*}
where $g_1$ is an arbitrary representative of an $\mathrm{EL}^{m-3}$-equivalence class of $G_1$ and $g_2$ is an arbitrary representative of an $\mathrm{EL}^{m-3}$-equivalence class of $G_2$, and have presented some variants of these representatives --- functions $g_{1}^{*},g_{2}^{*},g_{1}^{**}$ and $g_{2}^{**}$ (equations (1)-(4)).

Since the code $\mathcal{R}_{m-3}$ is linear, $\rho(\widehat{\mathcal{R}}_{m-3}) = \rho(\mathcal{R}_{m-3}) = m+2$ and the function $f$ is in $\widehat{\widehat{\mathcal{R}}}_{m-3}$ if and only if $f + \widehat{\mathcal{R}}_{m-3} = \widehat{\mathcal{R}}_{m-3}$.
Let us prove the metric regularity of $\mathcal{R}_{m-3}$ by proving that no functions other than those contained in $\mathcal{R}_{m-3}$ preserve the metric complement under addition.

Let $f\notin \mathcal{R}_{m-3}$ be an arbitrary function. Since $\widehat{\mathcal{R}}_{m-3}$ is a union of two $\mathrm{EL}^{m-3}$-equivalence classes, in order to show that $f + \widehat{\mathcal{R}}_{m-3} \neq \widehat{\mathcal{R}}_{m-3}$ it is enough to show that there exists a function $f'$ such that $f' \eleq{m-3} f$ and $f' + \widehat{\mathcal{R}}_{m-3} \neq \widehat{\mathcal{R}}_{m-3}$.

\textbf{Case 1.} Let $f\notin \mathcal{R}_{m-3}$ be a function of degree greater than $m-1$. Since $\mathrm{EL}^{m-3}$-equivalence preserves degree of functions with degree higher than $m-3$, any $g\in\widehat{\mathcal{R}}_{m-3}$ has degree $m-1$ or $m-2$ (like $g_{1}^{*}$ and $g_{2}^{*}$ respectively), while $f+g$ has higher degree and therefore cannot be equivalent to any of the functions from $\widehat{\mathcal{R}}_{m-3}$. Thus, functions of degree greater than $m-1$ cannot be in $\widehat{\widehat{\mathcal{R}}}_{m-3}$.

\textbf{Case 2.} Let $f\notin \mathcal{R}_{m-3}$ be a function of degree $m-1$. Any function of degree $m-1$ is trivially $\mathrm{EL}^{m-3}$-equivalent to a function with $\overline{x_m}$ as the only monomial of degree $(m-1)$, so
\begin{equation}
f \eleq{m-3} \overline{x_m} + x_m f_1 + f_2,
\end{equation}
where $f_1, f_2$ do not depend on $x_m$, $f_1$ is either zero or has degree $m-3$, while $f_2$ is either zero or has degree $m-2$.

\textbf{Case 2.1.} Assume that $f_1$ in (5) is nonzero. Then, from Lemma 8 and Lemma 6 it follows that
\begin{equation}
f \eleq{m-3} \overline{x_m} + x_m p_k + f_3 =: f'
\end{equation}
for some $k > 0$ and some $f_3$ of degree at most $m-2$ ($p_k, f_3$ do not depend on $x_m$). If we now sum $f'$ and $g_{2}^{**} \in \widehat{\mathcal{R}}_{m-3}$, we obtain:
\begin{equation*}
g_{2}^{**} + f' \eeq{m-3} \overline{x_m} + x_m (p_k + p_{\frac{m-1}{2}}) + f_3 \eleq{m-3} \overline{x_m} + x_m p_{\frac{m-1}{2}-k} + f_4,
\end{equation*}
where $f_4$ is a function of degree at most $m-2$, not depending on $x_m$, and the last equivalence is a simple variable renaming.

Let us denote this last function as $g'$. It has degree $m-1$ and therefore cannot be equivalent to the functions from $G_2$. It cannot be equivalent to the functions from $G_1$ either, because, by Lemma 8, any linear transformation of variables with matrix $\mathrm{D}$ which preserves $\overline{x_m}$ will act onto it in the following manner:
\begin{equation*}
g'\circ L_{\mathrm{D}} \eeq{m-3} \overline{x_m} + x_m (p_{\frac{m-1}{2}-k}\circ \bar{L}_{\bar{\mathrm{D}}}) + f_5,
\end{equation*}
where $f_5$ is some function of degree at most $m-2$ in the first $m-1$ variables. It is clear that no matrix $\bar{\mathrm{D}}$ can match the monomials of degree $m-2$ containing variable $x_m$ of the function $g'$ and of the function $g_{1}^{**}$, since $p_{\frac{m-1}{2}-k}$ is not equivalent to $p_{\frac{m-1}{2}}$. Thus, the function $g' = g_{2}^{**} + f'$ is not in $\widehat{\mathcal{R}}_{m-3}$, and therefore, if $f_1$ is nonzero, $f$ is not in $\widehat{\widehat{\mathcal{R}}}_{m-3}$.

\textbf{Case 2.2}
Assume that both $f_1$ and $f_2$ in (5) are zero. Then
\begin{equation*}
f \eleq{m-3} \overline{x_m} =: f'.
\end{equation*}
Using the transformation $L_{1m}: x_1 \leftarrow x_1+x_m$ (and removing the terms of degree less than $m-2$), the function
$g_1^* = \overline{x_m} + x_m g^{\star}$
transforms into
$g_1^* \circ L_{1m} \eeq{m-3} \overline{x_m} + \overline{x_1} + x_m g^{\star}$.

If we now sum $f'$ and $g_1^* \circ L_{1m}\in \widehat{\mathcal{R}}_{m-3}$ we will obtain the function $g' \eeq{m-3} \overline{x_1} + x_m g^{\star}$.
If we swap the variables $x_1$ and $x_m$ in it by another linear transformation and regroup terms, we will see that
\begin{equation*}
g' \eleq{m-3} \overline{x_m} + \sum\limits_{2\leqslant i < j\leqslant m-1} \overline{x_i x_j} + \sum\limits_{i=2}^{m-1} \overline{x_i x_m} \eleq{m-3} \overline{x_m} + x_m p_{\frac{m-3}{2}} + h
\end{equation*}
for some $h$ of degree at most $m-2$ in the first $m-1$ variables. By Lemma 8 and Lemma 6, this function cannot be equivalent to $g_{1}^{**}$ and it is not equivalent to $g_{2}^{*}$ by degree comparison. Therefore, $g'$ is not in $\widehat{\mathcal{R}}_{m-3}$, and hence $f$ is not in $\widehat{\widehat{\mathcal{R}}}_{m-3}$.

\textbf{Case 2.3} Assume that $f_1$ in (5) is zero and $f_2$ is nonzero. Then
\begin{equation*}
f \eleq{m-3} \overline{x_m} + f_2 =: f'.
\end{equation*}
Since $f_2$ does not contain the variable $x_m$, all terms of $f_2$ are of the form $\overline{x_i x_m}$ for some $i$. Without loss of generality (swapping variables among the first $m-1$ if needed) we can assume that $f_2$ contains $\overline{x_{m-1} x_m}$.
Renaming variables in $g_{2}^{**}$, we can transform it into:
\begin{equation*}
g_{2}^{**} \eleq{m-3} \overline{x_2 x_3} + \overline{x_4 x_5} + \ldots + \overline{x_{m-1} x_{m}}.
\end{equation*}
If we now add $f'$ and the function above, which belongs to $\widehat{\mathcal{R}}_{m-3}$, we will obtain the function $g' := \overline{x_m} + \sum\limits_{k=1}^{\frac{m-3}{2}} \overline{x_{2k} x_{2k+1}} + \sum\limits_{i\in I}\overline{x_i x_m}$, which is equivalent to
\begin{equation*}
g' \eleq{m-3} \overline{x_m} + x_m p_{\frac{m-3}{2}} + h
\end{equation*}
for some $h$ of degree at most $m-2$ in the first $m-1$ variables. By Lemma 8 and Lemma 6, this function  cannot be equivalent to $g_{1}^{**}$ and it is not equivalent to $g_2^*$ by degree comparison. Therefore, $g'$ is not in $\widehat{\mathcal{R}}_{m-3}$, and thus $f$ is not in $\widehat{\widehat{\mathcal{R}}}_{m-3}$.

\textbf{Case 3.} If $f\notin \mathcal{R}_{m-3}$ is a function of degree $m-2$, then, by arguments similar to the case of even $m$, $f$ is equivalent to $p_k$ (in $m$ variables) for some $k > 0$. Then
\begin{equation*}
p_k + g_{2}^{**} \eleq{m-3} p_{\frac{m-1}{2}-k}.
\end{equation*}
The function on the right is inequivalent to both $g_{2}^{**}$ (because $\frac{m-1}{2} \neq \frac{m-1}{2}-k$) and $g_1^*$ (by degree comparison), therefore $f\notin\widehat{\widehat{\mathcal{R}}}_{m-3}$.

Since all functions which are not in $\mathcal{R}_{m-3}$ have degree $m-2$ or higher, we have proven that none of them are in the second metric complement, and therefore $\mathcal{R}_{m-3}$ is metrically regular when $m$ is odd.

Factoring in the results from Section 7, we have proved the following
\begin{theorem}
	$\mathcal{R}_{m-3}$ is metrically regular for any $m\geqslant 3$.
\end{theorem}

\section{The Reed-Muller code $\mathcal{RM}(2,6)$} \label{RM26Section}

Let us consider one other special case. If we change the order of values in the value vectors of functions so that the first half of values corresponds to the values of the function when the last variable is set to $0$, and the other half corresponds to the values of the function when the last variable is set to $1$, then each Reed-Muller code (for $m>1$, $r>0$) can be inductively defined as follows:
\begin{equation*}
\mathcal{R}_{r,m} = \{(\mathbf{u},\mathbf{u+v}) : \mathbf{u}\in \mathcal{R}_{r,m-1}, \mathbf{v}\in \mathcal{R}_{r-1,m-1}\}.
\end{equation*}

In particular,
\begin{equation*}
\mathcal{R}_{2,6} = \{(\mathbf{u},\mathbf{u+v}) : \mathbf{u}\in \mathcal{R}_{2,5}, \mathbf{v}\in \mathcal{R}_{1,5}\}.
\end{equation*}

Since both $\mathcal{R}_{2,5}$ and $\mathcal{R}_{1,5}$ were shown to be metrically regular, this construction proves useful and allows us to establish the metric regularity of the code $\mathcal{R}_{2,6}$ as well. From now on, vectors in bold will represent value vectors of functions in $5$ variables (of length $32$), while value vectors of $6$-variable functions will be presented as pairs of value vectors of  $5$-variable functions. Additionaly, in this section we will use the notion of an \textit{automorphism} of a set, which will denote an isometric bijective mapping from the whole space to itself which maps the given set to itself.

First, let us establish one of the basic tools for the following investigations. Recall that $\rho(\mathcal{R}_{2,5}) = 6$ (Section 8), $\rho(\mathcal{R}_{1,5}) = 12$ \cite{BW72} and $\rho(\mathcal{R}_{2,6}) = 18$ \cite{SCH81}.

\begin{lemma}
	Let $\mathbf{(y,w)}$ be in $\widehat{\mathcal{R}}_{2,6}$. Then $\mathbf{y}\in \widehat{\mathcal{R}}_{2,5}$ and for any $\mathbf{u}\in \mathcal{R}_{2,5}$ such that $d(\mathbf{y},\mathbf{u}) = 6$ it holds $d(\mathbf{w+u}, \mathcal{R}_{1,5}) = 12$, i.e. $\mathbf{(w+u)} \in \widehat{\mathcal{R}}_{1,5}$.
\end{lemma}

\begin{proof}
	Assume that $\mathbf{y} \notin \widehat{\mathcal{R}}_{2,5}$, i.e. $d(\mathbf{y},\mathcal{R}_{2,5}) < 6$. Then there exists a vector $\mathbf{u}\in \mathcal{R}_{2,5}$ such that $d(\mathbf{y},\mathbf{u}) < 6$. From the inductive construction of $\mathcal{R}_{2,6}$ it follows that the distance between $\mathbf{(y,w)}$ and $\mathcal{R}_{2,6}$ is at most the distance between $\mathbf{y}$ and $\mathbf{u}$ plus the distance between $\mathbf{w}$ and $\mathbf{u}+\mathcal{R}_{1,5}$. The latter is in turn equal to $d(\mathbf{w+u},\mathcal{R}_{1,5})$, which is bounded by the covering radius of $\mathcal{R}_{1,5}$. Therefore, 
	\begin{equation*}
		d(\mathbf{(y,w)},\mathcal{R}_{2,6}) \leqslant d(\mathbf{y},\mathbf{u}) + d(\mathbf{w+u},\mathcal{R}_{1,5}) < 6 + 12 = 18.
	\end{equation*}
	This contradicts with  $\mathbf{(y,w)}\in \widehat{\mathcal{R}}_{2,6}$, hence $\mathbf{y}$ is in $\widehat{\mathcal{R}}_{2,5}$.
	
	The second part is now trivial: if  $\mathbf{u}\in \mathcal{R}_{2,5}$ is a vector such that $d(\mathbf{y},\mathbf{u}) = 6$, then the distance $d(\mathbf{w+u},\mathcal{R}_{1,5})$ has to achieve the maximum of $12$ in order for the vector $\mathbf{(y,w)}$ to be in the metric complement of $\mathcal{R}_{2,6}$.
\end{proof}

Let $\mathbf{(\tilde{u},\tilde{u}+\tilde{v})} \in \widehat{\widehat{\mathcal{R}}}_{2,6}$. We will prove that
$\mathbf{(\tilde{u},\tilde{u}+\tilde{v})}$ is in $\mathcal{R}_{2,6}$ in two steps: first we establish that $\mathbf{\tilde{u}}$ is in $\mathcal{R}_{2,5}$, then we prove that $\mathbf{\tilde{v}}$ is in $\mathcal{R}_{1,5}$.

Recall (Section 8) that $\widehat{\mathcal{R}}_{2,5} = \{g : g\eleq{2} g_1\} \cup \{g : g\eleq{2} g_2\}$, where $g_1$ and $g_2$ are some representatives of two $\mathrm{EL}^2$-equivalence classes. Let us denote
\begin{equation*}
	\widehat{\mathcal{R}}_{2,5}^1 := \{g : g\eleq{2} g_1\},\,\,\widehat{\mathcal{R}}_{2,5}^2 := \{g : g\eleq{2} g_2\}.
\end{equation*}
Then the following lemma is useful for proving the first half.

\begin{lemma}
	For each $i=1,2$ one of the following statements holds:
	\begin{enumerate}
		\item {$\forall \mathbf{y} \in \widehat{\mathcal{R}}_{2,5}^i\,\,\forall \mathbf{w}\in\mathbb{F}_2^{32}$ it holds $\mathbf{(y,w)}\notin \widehat{\mathcal{R}}_{2,6}$};
		\item {$\forall \mathbf{y} \in \widehat{\mathcal{R}}_{2,5}^i\,\,\exists \mathbf{w}\in\mathbb{F}_2^{32}$ such that $\mathbf{(y,w)}\in \widehat{\mathcal{R}}_{2,6}$};
	\end{enumerate}
\end{lemma}

\begin{proof}
	Assume that for some $i$ the second statement does not hold. Then, inverting it, we obtain
	\begin{equation*}
		\exists \mathbf{y^*} \in \widehat{\mathcal{R}}_{2,5}^i : \forall \mathbf{w}\in\mathbb{F}_2^{32} \text{ it holds } \mathbf{(y^*,w)}\notin \widehat{\mathcal{R}}_{2,6}
	\end{equation*}
	We will now prove that any $\mathbf{y} \in \widehat{\mathcal{R}}_{2,5}^i$ satisfies the claim of this statement and not just the vector $\mathbf{y^*}$. First, note that the statement ``$\forall \mathbf{w}\in\mathbb{F}_2^{32} \text{ it holds } \mathbf{(y^*,w)}\notin \widehat{\mathcal{R}}_{2,6}$'' is equivalent to the following:
	\begin{equation}
		\forall \mathbf{w}\in\mathbb{F}_2^{32}\,\,\exists \mathbf{u} \in \mathcal{R}_{2,5} : d(\mathbf{y^*,u}) + \min\limits_{\mathbf{v}\in\mathcal{R}_{1,5}} d(\mathbf{w+u,v}) < 18.
	\end{equation}
	
	Let $\mathbf{y}$ be an arbitrary vector from $\widehat{\mathcal{R}}_{2,5}^i$. Since all functions in $\widehat{\mathcal{R}}_{2,5}^i$ are $\mathrm{EL}^2$-equivalent, there exists a nonsingular linear transformation of variables $L$ and a function $g\in \mathcal{R}_{2,5}$ such that $f_\mathbf{y} = f_\mathbf{y^*} \circ L + g$. Let us denote as $\mathbf{g}$ the value vector of $g$ and as $\mathbf{L}$ the linear transformation on $\mathbb{F}_2^{32}$, corresponding to the transformation $L$ on functions. Then $\mathbf{y} = \mathbf{Ly^*} + \mathbf{g}$.
	
	Let us take an arbitrary vector $\mathbf{w}\in\mathbb{F}_2^{32}$ and let us denote $\mathbf{w_y} = \mathbf{L^{-1}(w+g)}$. Then, by (7), there exists a vector $\mathbf{u} \in \mathcal{R}_{2,5}$ such that
	\begin{equation}
    d(\mathbf{y^*,u}) + \min\limits_{\mathbf{v}\in\mathcal{R}_{1,5}} d(\mathbf{w_y+u,v}) < 18.
    \end{equation}
	
	Trivially, the transformation $\mathbf{L}$, as well as the addition of the function $\mathbf{g}$ are both automorphisms of $\mathbb{F}_2^{32}$. So let us apply $\mathbf{L}$ to all vectors being compared in (8) and add $\mathbf{g}$ to some of them, without changing the inequality:
	\begin{equation}
	d(\mathbf{Ly^*+g,Lu+g}) + \min\limits_{\mathbf{v}\in\mathcal{R}_{1,5}} d(\mathbf{Lw_y+Lu,Lv}) < 18.
	\end{equation}
	
	Note that the transformation $\mathbf{L}$ is also an automorphism of the codes $\mathcal{R}_{1,5}$ and $\mathcal{R}_{2,5}$. Let us denote $\mathbf{u_y}:=\mathbf{Lu+g}$. Then $\mathbf{u_y} \in \mathcal{R}_{2,5}$ and (9) can be transformed into
	\begin{equation}
	d(\mathbf{y,u_y}) + \min\limits_{\mathbf{v}\in\mathcal{R}_{1,5}} d(\mathbf{w+u_y,v}) < 18.
	\end{equation}
	
	Thus, for an arbitrary $\mathbf{w}\in\mathbb{F}_2^{32}$ we have found a vector $\mathbf{u_y} \in \mathcal{R}_{2,5}$ such that (10) holds. This means that (7) holds for the vector $\mathbf{y}$ that we have previously selected from $\widehat{\mathcal{R}}_{2,5}^i$. Since our selection was arbitrary, we have proved that the first statement of the lemma holds for the set $\widehat{\mathcal{R}}_{2,5}^i$ in question.
\end{proof}

\begin{proposition}
	Let $\mathbf{(\tilde{u},\tilde{u}+\tilde{v})} \in \widehat{\widehat{\mathcal{R}}}_{2,6}$. Then $\mathbf{\tilde{u}}\in\mathcal{R}_{2,5}$. 
\end{proposition}
\begin{proof}
Let us denote $Y := \{\mathbf{y}\in\mathbb{F}_2^{32}\,|\,\exists \mathbf{w}\in\mathbb{F}_2^{32} ; (\mathbf{y,w})\in\widehat{\mathcal{R}}_{2,6}\}$. From Lemma 9 it follows that $Y\subseteq \widehat{\mathcal{R}}_{2,5}$ and is nonempty. From Lemma 10 we can conclude that $Y$ can only coincide with one of the three sets: $\widehat{\mathcal{R}}_{2,5}^1$, $\widehat{\mathcal{R}}_{2,5}^2$ or $\widehat{\mathcal{R}}_{2,5}$.

Let $\mathbf{(\tilde{u},\tilde{u}+\tilde{v})} \in \widehat{\widehat{\mathcal{R}}}_{2,6}$. Then, as we know, $\mathbf{(\tilde{u},\tilde{u}+\tilde{v})} + \widehat{\mathcal{R}}_{2,6} = \widehat{\mathcal{R}}_{2,6}$.
This implies that $\mathbf{\tilde{u}} + Y = Y$. Since $\mathcal{R}_{2,5}$ is proven to be metrically regular, we know that only the vectors from $\mathcal{R}_{2,5}$ preserve its metric complement under addition. Following the proof of the metric regularity of the code $\mathcal{R}_{m-3,m}$ for $m$ odd (Subsection 8.3), it is easy to see that the same can be shown true for the sets $\widehat{\mathcal{R}}_{2,5}^1$ and $\widehat{\mathcal{R}}_{2,5}^2$ if they are considered separately one from another. Therefore, regardless of the contents of $Y$, only vectors from $\mathcal{R}_{2,5}$ preserve it under addition, and therefore $\mathbf{\tilde{u}}\in\mathcal{R}_{2,5}$.
\end{proof}

Recall from Section 3 that $\widehat{\mathcal{R}}_{1,5}$ is composed of $4$ $\mathrm{EA}$-equivalence classes: $\widehat{\mathcal{R}}_{1,5} = \bigcup_{i=1}^4 \widehat{\mathcal{R}}_{1,5}^i$. Similar to Lemma 10, the following statement holds:

\begin{lemma}
	For each $i=1,2,3,4$ one of the following statements holds:
	\begin{enumerate}
		\item {$\forall \mathbf{w'} \in \widehat{\mathcal{R}}_{1,5}^i\,\,\forall \mathbf{(y,w)}\in\widehat{\mathcal{R}}_{2,6}\,\,\forall\mathbf{u}\in\mathcal{R}_{2,5}\,\,(d\mathbf{(y,u)}=6 \rightarrow \mathbf{w+u \neq w'})$};
		\item {$\forall \mathbf{w'} \in \widehat{\mathcal{R}}_{1,5}^i\,\,\exists \mathbf{(y,w)}\in\widehat{\mathcal{R}}_{2,6}\,\,\exists\mathbf{u}\in\mathcal{R}_{2,5} : (d\mathbf{(y,u)}=6 \wedge \mathbf{w+u = w'})$};
	\end{enumerate}
\end{lemma}

\begin{proof}
	
	Assume that for some $i$ the second statement does not hold. Then there exists a vector $\mathbf{w^*} \in \widehat{\mathcal{R}}_{1,5}^i$ such that
	\begin{equation}
	\forall \mathbf{(y,w)}\in\widehat{\mathcal{R}}_{2,6}\,\,\forall\mathbf{u}\in\mathcal{R}_{2,5}\,\,(d\mathbf{(y,u)}=6 \rightarrow \mathbf{w+u \neq w^*})
	\end{equation}
	
	Let $\mathbf{w'}$ be an arbitrary vector from $\widehat{\mathcal{R}}_{1,5}^i$. Since all functions in $\widehat{\mathcal{R}}_{1,5}^i$ are $\mathrm{EA}$-equivalent, there exists a nonsingular affine transformation of variables $A$ and a function $g\in \mathcal{R}_{1,5}$ such that $f_\mathbf{w'} = f_\mathbf{w^*} \circ A + g$. Let us denote as $\mathbf{g}$ the value vector of $g$ and as $\mathbf{A}$ the linear transformation on $\mathbb{F}_2^{32}$, corresponding to the transformation $A$ on functions. Then $\mathbf{w'} = \mathbf{Aw^*} + \mathbf{g}$.
	
	Let $\mathbf{(y,w)}$ be an arbitrary vector from $\widehat{\mathcal{R}}_{2,6}$ and $\mathbf{u}$ be an arbitrary vector from $\mathcal{R}_{2,5}$. Since $\mathbf{A}$ is an automorphism of the codes $\mathcal{R}_{2,5}$ and $\mathcal{R}_{1,5}$, the vector $\mathbf{A^{-1}u}$ is in $\mathcal{R}_{2,5}$ and $\mathbf{(A,A)}\cdot \mathcal{R}_{2,6} = \mathcal{R}_{2,6}$. Since the addition of $\mathbf{g}$ is an automorphism of $\mathcal{R}_{2,5}$, $\mathcal{R}_{1,5}$ and $\mathbb{F}_2^{32}$, the vector $\mathbf{(A^{-1}y,A^{-1}(w+g))}$ is in $\widehat{\mathcal{R}}_{2,6}$. Hence, from (11) it follows that:
	\begin{equation}
		d\mathbf{(A^{-1}y,A^{-1}u)}=6 \rightarrow \mathbf{A^{-1}(w+g)+A^{-1}u \neq w^*}.
	\end{equation}
	
	Applying $\mathbf{A}$, (12) shows to be equivalent to the following statement:
	
	\begin{equation}
	d\mathbf{(y,u})=6 \rightarrow \mathbf{w+u \neq w'}.
	\end{equation}
	
	Hence, we have shown that for an arbitrary vector $\mathbf{w'}$ from $\widehat{\mathcal{R}}_{1,5}^i$, an arbitrary vector $\mathbf{(y,w)}$ from $\widehat{\mathcal{R}}_{2,6}$ and an arbitrary $\mathbf{u}$ from $\mathcal{R}_{2,5}$ (13) holds. This proves that the first statement holds for the class $\widehat{\mathcal{R}}_{1,5}^i$.
	
\end{proof}

The following result shows that any of the $\mathrm{EA}$-equivalence classes of the metric complement of $\mathcal{R}_{1,5}$ are also rather ``unstable'' when summed with a non-affine function:
\begin{lemma}
	For any $\mathbf{v}\notin \mathcal{R}_{1,5}$ and any $i=1,2,3,4$ there exists a vector $\mathbf{w}\in \widehat{\mathcal{R}}_{1,5}^i$ such that $\mathbf{v+w} \notin \widehat{\mathcal{R}}_{1,5}$.
\end{lemma}
\begin{proof}
	Like in Section 3, in order to prove the statement it is enough to show that for any $i=1,2,3,4$ and any $\mathrm{EA}$-equivalence class $C$ of $\mathbb{F}_2^{32}$ of even weight (other than $\mathcal{R}_{1,5}$) there exists
	a function $f\in C$ and a function $g\in \widehat{\mathcal{R}}_{1,5}^i$ such that $f+g \notin \widehat{\mathcal{R}}_{1,5}$. The proof for this can be found in the appendix in Tables~\ref{table:RM1514Cosets}-\ref{table:RM1528Cosets}.
\end{proof}

\begin{theorem}
	The code $\mathcal{R}_{2,6}$ is metrically regular.
\end{theorem}

\begin{proof}
	
Since any linear code is a subset of its second metric complement, we only need to prove that $\widehat{\widehat{\mathcal{R}}}_{2,6}\subseteq \mathcal{R}_{2,6}$. Let $\mathbf{(\tilde{u},\tilde{u}+\tilde{v})}$ be a vector from $\widehat{\widehat{\mathcal{R}}}_{2,6}$. We have already proved that $\mathbf{\tilde{u}}$ is in $\mathcal{R}_{2,5}$, therefore the vector $\mathbf{(0,\tilde{v})}$ is also in $\widehat{\widehat{\mathcal{R}}}_{2,6}$. Let us prove that $\mathbf{\tilde{v}}$ is in $\mathcal{R}_{1,5}$.	
	
Assume that $\mathbf{\tilde{v}}\notin\mathcal{R}_{1,5}$. Since (by Lemma 9) for an arbitrary $\mathbf{(y,w)}\in\widehat{\mathcal{R}}_{2,6}$ there exists a vector $\mathbf{u}\in\mathcal{R}_{2,5}$ such that $\mathbf{(w+u)}\in \widehat{\mathcal{R}}_{1,5}$, for some $i$ the second statement of Lemma 11 must hold. 

By Lemma 12, there exists a vector $\mathbf{w^*}\in \widehat{\mathcal{R}}_{1,5}^i$ such that $\mathbf{\tilde{v} + w^*} \notin \widehat{\mathcal{R}}_{1,5}$. 

By the second statement of Lemma 11, for this $\mathbf{w^*}$ there exists a vector $\mathbf{(y,w)}\in\widehat{\mathcal{R}}_{2,6}$ and a vector $\mathbf{u}\in\mathcal{R}_{2,5}$ such that $(d\mathbf{(y,u)}=6 \wedge \mathbf{w+u = w^*})$. Since $\mathbf{(0,\tilde{v})}\in\widehat{\widehat{\mathcal{R}}}_{2,6}$, $\mathbf{(y,w+\tilde{v})}$ is also in $\widehat{\mathcal{R}}_{2,6}$, and since $d\mathbf{(y,u)}=6$, by Lemma 9 the vector $\mathbf{w+\tilde{v}+u}$ is in $\widehat{\mathcal{R}}_{1,5}$. But $\mathbf{w+\tilde{v}+u} = \mathbf{\tilde{v} + w^*} \notin \widehat{\mathcal{R}}_{1,5}$, contradiction. Therefore, $\mathbf{\tilde{v}}\in\mathcal{R}_{1,5}$ and hence $\mathbf{(\tilde{u},\tilde{u}+\tilde{v})}\in \mathcal{R}_{2,6}$.
\end{proof}

\section{Conclusion}

In this paper we have established the metric regularity of the codes $\mathcal{RM}(1,5)$, $\mathcal{RM}(2,6)$ and of the codes $\mathcal{RM}(k,m)$ for $k\geqslant m-3$. Factoring in the result by Tokareva \cite{TOK10}, which proves the metric regularity of $\mathcal{RM}(1,m)$ for even $m$, all infinite families of Reed-Muller codes with known covering radius are covered. The only other Reed-Muller codes with known covering radius, metric regularity of which has not been yet established, are $\mathcal{RM}(1,7)$ and $\mathcal{RM}(2,7)$.  Given these results, we formulate the following
\begin{conj*}
    All Reed-Muller codes $\mathcal{RM}(k,m)$ are metrically regular.
\end{conj*}
The availability of the coset weight distributionfor allowed us to consider the code $\mathcal{RM}(1,5)$, and the fact that the covering radius of $\mathcal{RM}(2,6)$ attains an upper bound given by the $(\mathbf{u},\mathbf{u+v})$ construction \cite{SCH81} allowed us to establish its metric regularity even without describing the metric complement. However, the codes $\mathcal{RM}(1,7)$ and $\mathcal{RM}(2,7)$ are much harder to consider because of the lack of similar regularities, the larger number of variables, the larger covering radius and the unconstructive nature of the results which describe their covering radius.

I would like to thank Natalia Tokareva, Alexander Kutsenko and the collective of the Selmer Center of the University in Bergen for the inspiration and helpful remarks during the development of this work.


\begin{thebibliography}{}
%
%
\bibitem {BW72}
Berlekamp E., Welch L.: Weight distributions of the cosets of the $(32, 6)$ Reed-Muller code. IEEE Transactions on Information Theory. \textbf{18}(1), 203--207 (1972).

\bibitem {COH97}
Cohen G., Honkala I., Litsyn S., Lobstein A.: Covering codes. Elsevier. \textbf{54}, (1997).

\bibitem {HOU96}
Hou X. D.: Covering Radius of the Reed-Muller code $R(1, 7)$ -- A Simpler Proof. Journal of Combinatorial Theory, Series A. \textbf{74}(2), 337--341 (1996).

\bibitem {KOL17}
Kolomeec N.: The graph of minimal distances of bent functions and its properties. Designs, Codes and Cryptography. \textbf{85}(3), 395--410 (2017).

\bibitem {KUT19}
Kutsenko A.: Metrical properties of self-dual bent functions. Designs, Codes and Cryptography (2019). doi:10.1007/s10623-019-00678-x

\bibitem {MCL79}
McLoughlin A. M.: The Covering Radius of the $(m-3)$-rd Order Reed Muller Codes and a Lower Bound on the $(m-4)$-th Order Reed Muller Codes. SIAM Journal on Applied Mathematics. \textbf{37}(2), 419--422 (1979).

\bibitem{MES16}
Mesnager S.: Bent Functions: Fundamentals and Results. Springer International Publishing, (2016).

\bibitem {MYK80}
Mykkeltveit J.: The covering radius of the $(128, 8)$ Reed-Muller code is $56$. IEEE Transactions on Information Theory. \textbf{26}(3), 359--362 (1980).

\bibitem {NEU92}
Neumaier A.: Completely regular codes. Discrete mathematics. \textbf{106}, 353--360 (1992).

\bibitem {OBL16}
Oblaukhov A. K.: Metric complements to subspaces in the Boolean cube. Journal of Applied and Industrial Mathematics. \textbf{10}(3), 397--403 (2016).

\bibitem {OBL18}
Oblaukhov A. K.: Maximal metrically regular sets. Siberian Electronic Mathematical Reports. \textbf{15}, 1842--1849 (2018).

\bibitem {OBL19}
Oblaukhov A.: A lower bound on the size of the largest metrically regular subset of the Boolean cube. Cryptography and Communications. \textbf{11}(4), 777--791 (2019).

\bibitem {ROT76}
Rothaus O. S.: On ``bent'' functions. Journal of Combinatorial Theory, Series A. \textbf{20}(3), 300--305 (1976).

\bibitem {SCH81}
Schatz J.: The second order Reed-Muller code of length $64$ has covering radius $18$. IEEE Transactions on Information Theory. \textbf{27}(4), 529--530 (1981).

\bibitem {STA18}
Stanica P., Sasao T., Butler J. T.: Distance duality on some classes of Boolean functions. Journal of Combinatorial Mathematics and Combinatorial Computing. 2018.

\bibitem {TOK10}
Tokareva N. N.: The group of automorphisms of the set of bent functions. Discrete Mathematics and Applications. \textbf{20}(5--6), 655--664 (2010).

\bibitem {TOK12}
Tokareva N.: Duality between bent functions and affine functions. Discrete Mathematics. \textbf{312}(3), 666--670 (2012).

\bibitem {TOK15}
Tokareva N.: Bent functions: results and applications to cryptography. Academic Press, (2015).

\bibitem {WAN19}
Wang Q.: The covering radius of the Reed–Muller code $RM(2,7)$ is $40$. Discrete Mathematics. \textbf{342}(12), Article 111625 (2019).

\end{thebibliography}



\section*{Appendix}

Tables~\ref{table:RM1514Cosets}-\ref{table:RM1528Cosets} show that for any $\mathrm{EA}$-equivalence class $\widehat{\mathcal{R}}_{1,5}^i$ of $\widehat{\mathcal{R}}_{1,5}$ and for each $\mathrm{EA}$-equivalence class $C$ of $\mathbb{F}_2^{32}$ there exists a function $f\in C$ and a function $g\in \widehat{\mathcal{R}}_{1,5}^i$ such that $f+g$ does not belong to $\widehat{\mathcal{R}}_{1,5}$. Note that, if this function $f$ is not in $\widehat{\mathcal{R}}_{1,5}$ and $f+g$ belongs to a class $C'$, then we do not have to search for a function with such properties in the class $C'$ since $(f+g) + g = f$ does not belong to $\widehat{\mathcal{R}}_{1,5}$ --- this is why some rows in the following tables are skipped.

Notations in Tables~\ref{table:RM1514Cosets}-\ref{table:RM1528Cosets} are the same as in Table~\ref{table:RM15Cosets} (see Section 3). The second column of Table~\ref{table:RM1528Cosets} contains ``canonical'' representatives for each $\mathrm{EA}$-equivalence class, as they were obtained in the paper \cite{BW72} by Berlekamp and Welch. In other columns and tables, some representatives are changed by either simple variable swaps or more complex transformations. These more complex transformations are marked with an asterisk and explained below for each table, along with other clarifications. Hereafter ``$i{\leftarrow}i+j$'' stands for ``$x_i{\leftarrow}x_i + x_j$'', while two-way arrows denote variable swapping; all transformations are applied consecutively.

\textbf{Table~\ref{table:RM1514Cosets}:} Representatives $f$ for classes $7$, $9$, $10$ and $22$ (column $2$) are obtained from ``canonical'' using the following transformations:\\
$(7)\,3{\leftarrow}3{+}0;\,(9)\,4{\leftarrow}4{+}3{+}0;\,(10)\,1{\leftarrow}1{+}0;\,(22)\,4{\leftrightarrow}5;1{\leftrightarrow}3;$ \\
Functions $g$ from $\widehat{\mathcal{R}}_{1,5}^1$ (third column) are obtained from ``canonical'' using transformations: \\
$2345{+}123{+}24{+}35 \circ (2{\leftarrow}2{+}0) = 2345{+}345{+}123{+}13{+}24{+}35;$\\
$2345{+}123{+}24{+}35 \circ (5{\leftarrow}5{+}0) = 2345{+}234{+}123{+}24{+}35;$\\
$2345{+}123{+}24{+}35 \circ (1{\leftarrow}1{+}0) = 2345{+}123{+}24{+}35{+}23;$

Transformations which produce function in column $5$ from $h$ in column $4$:\\
$(1)\,2{\leftarrow}2{+}0;4{\leftarrow}4{+}0;1{\leftrightarrow}3;\,(2)\,2{\leftarrow}2{+}0;4{\leftarrow}4{+}0;1{\leftarrow}1{+}4;3{\leftarrow}3{+}0;5{\leftarrow}5{+}2;1{\leftrightarrow}3;$\\
$2{\leftrightarrow}4;3{\leftrightarrow}5;\,(3)\,1{\leftrightarrow}3;4{\leftrightarrow}5;\,(5)\,3{\leftarrow}3{+}0;1{\leftrightarrow}2;\,(6)\,3{\leftarrow}3{+}0;1{\leftrightarrow}3;3{\leftrightarrow}4;4{\leftrightarrow}5;$ \\
$(7)\,5{\leftarrow}5{+}0;1{\leftrightarrow}5;3{\leftrightarrow}4;\,(8)\,1{\leftrightarrow}3;2{\leftrightarrow}5;\,(9)\,4{\leftarrow}4{+}0;1{\leftarrow}1{+}2;1{\leftrightarrow}4;2{\leftrightarrow}5;$ \\
$(10)\,4{\leftarrow}4{+}3;2{\leftarrow}2{+}5;1{\leftrightarrow}4;\,(11)\,1{\leftarrow}1{+}4;\,(12)\,1{\leftarrow}1{+}2;2{\leftrightarrow}4;\,(13)\,3{\leftarrow}3{+}0;4{\leftarrow}4{+}0;$ \\
$1{\leftrightarrow}4;2{\leftrightarrow}5;4{\leftrightarrow}5;\,(15)\,3{\leftarrow}3{+}0;1{\leftrightarrow}4;2{\leftrightarrow}3;4{\leftrightarrow}5;\,(16)\,1{\leftarrow}1{+}0;3{\leftarrow}3{+}0;1{\leftrightarrow}4;2{\leftrightarrow}3;$ \\
$4{\leftrightarrow}5;\,(17)\,1{\leftarrow}1{+}0;1{\leftrightarrow}5;3{\leftrightarrow}4;2{\leftrightarrow}5;\,(18)\,2{\leftrightarrow}4;3{\leftrightarrow}5;\,(20)\,4{\leftarrow}4{+}3{+}0;5{\leftarrow}5{+}2{+}0;$ \\
$(23)\,2{\leftrightarrow}4;3{\leftrightarrow}5;5{\leftarrow}5{+}2{+}0;1{\leftarrow}1{+}0;4{\leftarrow}4{+}3{+}0;\,(24)\,2{\leftrightarrow}4;3{\leftrightarrow}5;5{\leftarrow}5{+}2{+}0;$ \\
$1{\leftarrow}1{+}0;4{\leftarrow}4{+}3{+}0;\,(26)\,2{\leftrightarrow}4;3{\leftrightarrow}5;$ \\

\textbf{Table~\ref{table:RM1522Cosets}:} Functions $g$ from $\widehat{\mathcal{R}}_{1,5}^2$ (third column) are obtained from ``canonical'' using variable swaps. Transformations which produce function in column $5$ from $h$ in column $4$:\\
$(1)\,2{\leftrightarrow}3;\,(2)\,4{\leftarrow}4{+}2{+}0;2{\leftrightarrow}3;\,(3)\,1{\leftarrow}1{+}0;2{\leftrightarrow}3;\,(4)\,3{\leftarrow}3{+}2{+}0;1{\leftarrow}1{+}2{+}3;$ \\
$(7)\,4{\leftarrow}4{+}2{+}0;2{\leftrightarrow}4;3{\leftrightarrow}5;\,(8)\,2{\leftrightarrow}4;\,(10)\,2{\leftarrow}2{+}4{+}0;2{\leftrightarrow}4;3{\leftrightarrow}5;\,(11)\,2{\leftarrow}2{+}5{+}0;$ \\
$(13)\,2{\leftarrow}2{+}5{+}0;5{\leftarrow}5{+}3{+}0;3{\leftrightarrow}5;\,(15)\,2{\leftrightarrow}4;3{\leftrightarrow}5;\,(16)\,1{\leftarrow}1{+}0;2{\leftrightarrow}4;3{\leftrightarrow}5;$ \\
$(17)\,1{\leftarrow}1{+}0;2{\leftrightarrow}5;3{\leftrightarrow}4;\,(18)\,2{\leftrightarrow}4;3{\leftrightarrow}5;\,(19)\,3{\leftarrow}3{+}0;1{\leftrightarrow}5;\,(21)\,1{\leftrightarrow}5;$ \\
$(23)\,5{\leftarrow}5{+}0;1{\leftrightarrow}5;2{\leftrightarrow}4;3{\leftrightarrow}5;\,(24)\,5{\leftarrow}5{+}0;3{\leftarrow}3{+}5;2{\leftrightarrow}4;3{\leftrightarrow}5;\,(25)\,4{\leftarrow}4{+}0;2{\leftrightarrow}4;$ \\
$3{\leftrightarrow}5;\,(26)\,4{\leftarrow}4{+}0;2{\leftarrow}2{+}5;2{\leftrightarrow}4;3{\leftrightarrow}5;\,(27)\,1{\leftrightarrow}2;4{\leftrightarrow}5;$ \\

\textbf{Table~\ref{table:RM1526Cosets}:} Representatives $f$ for classes $4$ and $9$ (column $2$) are obtained from ``canonical'' using the following transformations: $(4)\,3{\leftarrow}3{+}0;\,(9)\,1{\leftarrow}1{+}2;$ \\
Functions $g$ from $\widehat{\mathcal{R}}_{1,5}^3$ (third column) are obtained from ``canonical'' using transformations: \\
$123{+}145{+}23{+}24{+}35 \circ (1{\leftarrow}1{+}2) = 123{+}145{+}245{+}24{+}35;$\\
$123{+}145{+}245{+}24{+}35 \circ (3{\leftarrow}3{+}0) = 123{+}145{+}245{+}24{+}35{+}12;$\\
$123{+}145{+}245{+}24{+}35{+}12 \circ (4{\leftrightarrow}5) = 123{+}145{+}245{+}25{+}34{+}12;$\\
$123{+}145{+}245{+}24{+}35{+}12 \circ (2{\leftrightarrow}4;3{\leftrightarrow}5) = 123{+}145{+}234{+}24{+}35{+}14;$\\
$123{+}145{+}23{+}24{+}35 \circ (2{\leftrightarrow}4;3{\leftrightarrow}5) = 123{+}145{+}45{+}24{+}35;$

Transformations which produce function in column $5$ from $h$ in column $4$:\\
$(1)\,3{\leftarrow}3{+}0;\,(2)\,3{\leftarrow}3{+}0;\,(4)\,1{\leftarrow}1{+}0;\,(5)\,1{\leftarrow}1{+}2;1{\leftarrow}1{+}0;\,(6)\,2{\leftarrow}2{+}5{+}0;$ \\
$3{\leftarrow}3{+}4{+}0;1{\leftarrow}1{+}0;2{\leftrightarrow}4;3{\leftrightarrow}5;\,(7)\,3{\leftarrow}3{+}4;1{\leftarrow}1{+}4;2{\leftrightarrow}4;3{\leftrightarrow}5;\,(8)\,2{\leftarrow}2{+}5{+}0;$ \\
$2{\leftrightarrow}4;3{\leftrightarrow}5;\,(9)\,5{\leftarrow}5{+}0;2{\leftarrow}2{+}5{+}0;2{\leftrightarrow}4;3{\leftrightarrow}5;\,(11)\,3{\leftarrow}3{+}2;2{\leftrightarrow}4;3{\leftrightarrow}5;2{\leftrightarrow}3;$ \\
$(14)\,2{\leftrightarrow}4;3{\leftrightarrow}5;\,(15)\,5{\leftarrow}5{+}2{+}0;3{\leftrightarrow}4;\,(16)\,3{\leftrightarrow}4;\,(17)\,4{\leftarrow}4{+}3{+}0;\,(19)\,1{\leftarrow}1{+}0;$ \\
$3{\leftarrow}3{+}4;2{\leftarrow}2{+}5;2{\leftrightarrow}4;3{\leftrightarrow}5;\,(21)\,3{\leftarrow}3{+}0;1{\leftrightarrow}4;2{\leftrightarrow}3;4{\leftrightarrow}5;\,(22)\,5{\leftarrow}5{+}0;2{\leftarrow}2{+}4;$ \\
$2{\leftrightarrow}4;3{\leftrightarrow}5;\,(23)\,5{\leftarrow}5{+}0;1{\leftrightarrow}5;3{\leftrightarrow}4;\,(24)\,5{\leftarrow}5{+}0;5{\leftarrow}5{+}2{+}0;1{\leftrightarrow}5;3{\leftrightarrow}4;$ \\
$(27)\,5{\leftarrow}5{+}2;1{\leftarrow}1{+}0;3{\leftarrow}3{+}4;$ \\

\textbf{Table~\ref{table:RM1528Cosets}:} Functions $g$ from $\widehat{\mathcal{R}}_{1,5}^4$ (third column) are obtained from ``canonical'' using variable swaps. Transformations which produce function in column $5$ from $h$ in column $4$:\\
$(2)\,2{\leftrightarrow}4;\,(7)\,2{\leftrightarrow}3;\,(9)\,2{\leftrightarrow}3;4{\leftrightarrow}5;\,(16)\,1{\leftarrow}1{+}0;\,(17)\,1{\leftarrow}1{+}0;\,(19)\,3{\leftarrow}3{+}0;$ \\
$(21)\,1{\leftrightarrow}2;4{\leftrightarrow}5;\,(23)\,1{\leftarrow}1{+}0;\,(24)\,1{\leftarrow}1{+}0;\,(25)\,2{\leftrightarrow}3;4{\leftrightarrow}5;\,(27)\,1{\leftrightarrow}3;2{\leftrightarrow}4;$ \\

\begin{table}[h]
	\begin{adjustbox}{width=\columnwidth,center}
		\begin{tabular}{|l|l|l|l|l|l|}
			\hline
			No    & Representative $f$           & $g$ from $\widehat{\mathcal{R}}_{1,5}^1 (14)$ & Sum $h=f{+}g$                     & $h$ is equal to   & $C(h)$ \\ \hline
			
			$0$   & $0$                          & ---                                      & ---                               & ---               & --- \\ \hline
			$1$   & $2345$                       & $2345{+}345{+}123{+}13{+}24{+}35$        & $123{+}345{+}13{+}24{+}35$        & $123{+}145{+}24$  & $25$  \\ \hline
			$2$   & $2345{+}12$                  & $2345{+}345{+}123{+}13{+}24{+}35$        & $123{+}345{+}12{+}13{+}24{+}35$   & $123{+}145{+}23$  & $24$ \\ \hline
			$3$   & $2345{+}24$                  & $2345{+}123{+}24{+}35$                   & $123{+}35$                        & $123{+}14$        & $21$  \\ \hline
			$4$   & $2345{+}24{+}35$             & $2345{+}123{+}24{+}35$                   & $123$                             & $\leftarrow$      & $19$  \\ \hline
			$5$   & $2345{+}12{+}35$             & $2345{+}123{+}24{+}35$                   & $123{+}12{+}24$                   & $123{+}14$        & $21$  \\ \hline
			$6$   & $2345{+}123$                 & $2345{+}234{+}123{+}24{+}35$             & $234{+}24{+}35$                   & $123{+}14$        & $21$  \\ \hline
			$7$   & $2345{+}245{+}123^*$         & $2345{+}123{+}24{+}35$                   & $245{+}24{+}35$                   & $123{+}14$        & $21$  \\ \hline
			$8$   & $2345{+}123{+}24$            & $2345{+}123{+}24{+}35$                   & $35$                              & $12$              & $27$  \\ \hline
			$9$   & $2345{+}123{+}14{+}13^*$     & $2345{+}345{+}123{+}13{+}24{+}35$        & $345{+}14{+}24{+}35$              & $123{+}14$        & $21$  \\ \hline
			$10$  & $2345{+}123{+}45{+}23^*$     & $2345{+}123{+}24{+}35$                   & $23{+}24{+}35{+}45$               & $12$              & $27$  \\ \hline
			$11$  & $2345{+}123{+}12{+}35$       & $2345{+}123{+}24{+}35$                   & $12{+}24$                         & $12$              & $27$  \\ \hline
			$12$  & $2345{+}123{+}14{+}35$       & $2345{+}123{+}24{+}35$                   & $14{+}24$                         & $12$              & $27$  \\ \hline
			$13$  & $2345{+}123{+}13{+}45$       & $2345{+}345{+}123{+}13{+}24{+}35$        & $345{+}24{+}35{+}45$              & $123{+}14$        & $21$  \\ \hline
			$14^1$& $2345{+}123{+}24{+}35$       & $2345{+}123{+}24{+}35$                   & $0$                               & $\leftarrow$      & $0$  \\ \hline
			$15$  & $2345{+}123{+}145$           & $2345{+}234{+}123{+}24{+}35$             & $145{+}234{+}24{+}35$             & $123{+}145{+}24$  & $25$  \\ \hline
			$16$  & $2345{+}123{+}145{+}45$      & $2345{+}234{+}123{+}24{+}35$             & $145{+}234{+}24{+}35{+}45$        & $123{+}145{+}24$  & $25$  \\ \hline
			$17$  & $2345{+}123{+}145{+}24{+}45$ & $2345{+}123{+}24{+}35$                   & $145{+}45{+}35$                   & $123{+}14$        & $21$  \\ \hline
			$18$  & $2345{+}123{+}145{+}24{+}35$ & $2345{+}123{+}24{+}35$                   & $145$                             & $123$             & $19$  \\ \hline
			$19$  & $123$                        & ---                                      & ---                               & ---               & ---  \\ \hline
			$20$  & $123{+}45$                   & $2345{+}123{+}24{+}35$                   & $2345{+}24{+}35{+}45$             & $2345{+}23{+}45$  & $4$  \\ \hline
			$21$  & $123{+}14$                   & ---                                      & ---                               & ---               & ---  \\ \hline
			$22^2$& $123{+}24{+}35^*$            & $2345{+}123{+}24{+}35$                   & $2345$                            & $\leftarrow$      & $1$  \\ \hline
			$23$  & $123{+}145$                  & $2345{+}123{+}24{+}35{+}23$              & $2345{+}145{+}24{+}35{+}23$       & $2345{+}123{+}45$ & $10$  \\ \hline
			$24$  & $123{+}145{+}23$             & $2345{+}123{+}24{+}35$                   & $2345{+}145{+}24{+}35{+}23$       & $2345{+}123{+}45$      & $10$  \\ \hline
			$25$  & $123{+}145{+}24$             & ---                                      & ---                               & ---               & ---  \\ \hline
			$26^3$& $123{+}145{+}23{+}24{+}35$   & $2345{+}123{+}24{+}35$                   & $2345{+}145{+}23$                 & $2345{+}123{+}45$ & $10$  \\ \hline
			$27$  & $12$                         & ---                                      & ---                               & ---               & ---  \\ \hline
			$28^4$& $24{+}35$                    & $2345{+}123{+}24{+}35$                   & $2345{+}123$                      & $\leftarrow$      & $6$  \\ \hline
		\end{tabular}
	\end{adjustbox}
	\caption{Proof of Lemma 12 for the class $\widehat{\mathcal{R}}_{1,5}^1$.}\label{table:RM1514Cosets}
\end{table}

\begin{table}[]
	\begin{adjustbox}{width=\columnwidth,center}
		\begin{tabular}{|l|l|l|l|l|l|}
			\hline
			No    & Representative $f$           & $g$ from $\widehat{\mathcal{R}}_{1,5}^2 (22)$ & Sum $h=f{+}g$                & $h$ is equal to         & $C(h)$ \\ \hline
			
			$0$   & $0$                          & ---                                      & ---                          & ---                     & --- \\ \hline
			$1$   & $2345$                       & $123{+}14{+}25$                          & $2345{+}123{+}14{+}25$       & $2345{+}123{+}14{+}35$  & $12$  \\ \hline
			$2$   & $2345{+}12$                  & $123{+}14{+}25$                          & $2345{+}123{+}12{+}14{+}25$  & $2345{+}123{+}14{+}35$  & $12$ \\ \hline
			$3$   & $2345{+}23$                  & $123{+}14{+}25$                          & $2345{+}123{+}23{+}14{+}25$  & $2345{+}123{+}14{+}35$  & $12$  \\ \hline
			$4$   & $2345{+}25{+}34$             & $123{+}14{+}25$                          & $2345{+}123{+}14{+}34$       & $2345{+}123{+}14$       & $9$  \\ \hline
			$5$   & $2345{+}14{+}25$             & $123{+}14{+}25$                          & $2345{+}123$                 & $\leftarrow$            & $6$  \\ \hline
			$6$   & $2345{+}123$                 & ---                                      & ---                          & ---                     & $21$  \\ \hline
			$7$   & $2345{+}123{+}12$            & $123{+}14{+}25$                          & $2345{+}12{+}14{+}25$        & $2345{+}12{+}34$        & $5$  \\ \hline
			$8$   & $2345{+}123{+}25$            & $123{+}14{+}25$                          & $2345{+}14$                  & $2345{+}12$             & $2$  \\ \hline
			$9$   & $2345{+}123{+}14$            & ---                                      & ---                          & ---                     & $21$  \\ \hline
			$10$  & $2345{+}123{+}45$            & $123{+}14{+}25$                          & $2345{+}14{+}25{+}45$        & $2345{+}12{+}34$        & $5$  \\ \hline
			$11$  & $2345{+}123{+}12{+}34$       & $123{+}15{+}34$                          & $2345{+}12{+}15$             & $2345{+}12$             & $2$  \\ \hline
			$12$  & $2345{+}123{+}14{+}35$       & ---                                      & ---                          & ---                     & $27$  \\ \hline
			$13$  & $2345{+}123{+}12{+}45$       & $123{+}15{+}34$                          & $2345{+}12{+}15{+}45{+}34$   & $2345{+}12{+}34$        & $5$  \\ \hline
			$14^1$& $2345{+}123{+}24{+}35$       & $123{+}24{+}35$                          & $2345$                       & $\leftarrow$            & $1$  \\ \hline
			$15$  & $2345{+}123{+}145$           & $123{+}14{+}25$                          & $2345{+}145{+}14{+}25$       & $2345{+}12{+}34$        & $11$  \\ \hline
			$16$  & $2345{+}123{+}145{+}45$      & $123{+}14{+}25$                          & $2345{+}145{+}14{+}25{+}45$  & $2345{+}12{+}34$        & $11$  \\ \hline
			$17$  & $2345{+}123{+}145{+}24{+}45$ & $123{+}24{+}35$                          & $2345{+}145{+}35{+}45$       & $2345{+}123{+}24$       & $8$  \\ \hline
			$18$  & $2345{+}123{+}145{+}24{+}35$ & $123{+}24{+}35$                          & $2345{+}145$                 & $2345{+}123$            & $6$  \\ \hline
			$19$  & $123{+}235$                  & $123{+}14{+}25$                          & $235{+}14{+}25$              & $123{+}45$              & $20$  \\ \hline
			$20$  & $123{+}45$                   & ---                                      & ---                          & ---                     & ---  \\ \hline
			$21$  & $123{+}14$                   & $123{+}14{+}25$                          & $25$                         & $12$                    & $27$  \\ \hline
			$22^2$& $123{+}14{+}25$              & $123{+}14{+}25$                          & $0$                          & $\leftarrow$            & $0$  \\ \hline
			$23$  & $123{+}145$                  & $123{+}14{+}25$                          & $145{+}14{+}25$              & $123{+}14$              & $21$  \\ \hline
			$24$  & $123{+}145{+}23$             & $123{+}14{+}25$                          & $145{+}14{+}25{+}23$         & $123{+}45$              & $20$  \\ \hline
			$25$  & $123{+}145{+}24$             & $123{+}15{+}24$                          & $145{+}15$                   & $123$                   & $19$  \\ \hline
			$26^3$& $123{+}145{+}23{+}24{+}35$   & $123{+}15{+}24$                          & $145{+}15{+}23{+}35$         & $123{+}45$              & $20$  \\ \hline
			$27$  & $14$                         & $123{+}14{+}25$                          & $123{+}25$                   & $123{+}14$              & $21$  \\ \hline
			$28^4$& $14{+}25$                    & $123{+}14{+}25$                          & $123$                        & $\leftarrow$            & $19$  \\ \hline
		\end{tabular}
	\end{adjustbox}
	\caption{Proof of Lemma 12 for the class $\widehat{\mathcal{R}}_{1,5}^2$.}\label{table:RM1522Cosets}
\end{table}

\begin{table}[]
	\begin{adjustbox}{width=\columnwidth,center}
		\begin{tabular}{|l|l|l|l|l|l|}
			\hline
			No    & Representative $f$           & $g$ from $\widehat{\mathcal{R}}_{1,5}^3 (26)$ & Sum $h=f{+}g$                           & $h$ is equal to              & $C(h)$ \\ \hline
			
			$0$   & $0$                          & ---                                      & ---                                     & ---                          & --- \\ \hline
			$1$   & $2345$                       & $123{+}145{+}245{+}24{+}35{+}12$         & $2345{+}123{+}145{+}245{+}24{+}35{+}12$ & $2345{+}123{+}145{+}24{+}35$ & $18$  \\ \hline
			$2$   & $2345{+}12$                  & $123{+}145{+}245{+}24{+}35$              & $2345{+}123{+}145{+}245{+}24{+}35{+}12$ & $2345{+}123{+}145{+}24{+}35$ & $18$ \\ \hline
			$3$   & $2345{+}23$                  & $123{+}145{+}23{+}24{+}35$               & $2345{+}123{+}145{+}24{+}35$            & $\leftarrow$                 & $18$  \\ \hline
			$4$   & $2345{+}245{+}23{+}45^*$     & $123{+}145{+}245{+}24{+}35$              & $2345{+}123{+}145{+}23{+}24{+}35{+}45$  & $2345{+}123{+}145{+}24{+}35$ & $18$  \\ \hline
			$5$   & $2345{+}12{+}35$             & $123{+}145{+}245{+}24{+}35{+}12$         & $2345{+}123{+}145{+}245{+}24$           & $2345{+}123{+}145{+}24{+}45$ & $17$  \\ \hline
			$6$   & $2345{+}123$                 & $123{+}145{+}23{+}24{+}35$               & $2345{+}145{+}23{+}24{+}35$             & $2345{+}123{+}45$            & $10$  \\ \hline
			$7$   & $2345{+}123{+}12$            & $123{+}145{+}245{+}24{+}35$              & $2345{+}145{+}245{+}24{+}35{+}12$       & $2345{+}123{+}35{+}14$       & $12$  \\ \hline
			$8$   & $2345{+}123{+}24$            & $123{+}145{+}23{+}24{+}35$               & $2345{+}145{+}23{+}35$                  & $2345{+}123{+}45$            & $10$  \\ \hline
			$9$   & $2345{+}123{+}14{+}23{+}24^*$& $123{+}145{+}234{+}24{+}35{+}14$         & $2345{+}145{+}234{+}23{+}35$            & $2345{+}123{+}12{+}45$       & $13$  \\ \hline
			$10$  & $2345{+}123{+}45$            & ---                                      & ---                                     & ---                          & ---  \\ \hline
			$11$  & $2345{+}123{+}12{+}34$       & $123{+}145{+}245{+}25{+}34{+}12$         & $2345{+}145{+}245{+}25$                 & $2345{+}123{+}24$            & $8$  \\ \hline
			$12$  & $2345{+}123{+}14{+}35$       & ---                                      & ---                                     & ---                          & ---  \\ \hline
			$13$  & $2345{+}123{+}12{+}45$       & ---                                      & ---                                     & ---                          & ---  \\ \hline
			$14^1$& $2345{+}123{+}24{+}35$       & $123{+}145{+}23{+}24{+}35$               & $2345{+}145{+}23$                       & $2345{+}123{+}45$            & $10$  \\ \hline
			$15$  & $2345{+}123{+}145$           & $123{+}145{+}23{+}24{+}35$               & $2345{+}23{+}24{+}35$                   & $2345{+}23{+}45$             & $4$  \\ \hline
			$16$  & $2345{+}123{+}145{+}45$      & $123{+}145{+}45{+}24{+}35$               & $2345{+}24{+}35$                        & $2345{+}23{+}45$             & $4$  \\ \hline
			$17$  & $2345{+}123{+}145{+}24{+}45$ & $123{+}145{+}23{+}24{+}35$               & $2345{+}23{+}35{+}45$                   & $2345{+}23{+}45$             & $4$  \\ \hline
			$18$  & $2345{+}123{+}145{+}24{+}35$ & ---                                      & ---                                     & ---                          & ---  \\ \hline
			$19$  & $123$                        & $123{+}145{+}23{+}24{+}35$               & $145{+}23{+}24{+}35$                    & $123{+}45$                   & $20$  \\ \hline
			$20$  & $123{+}45$                   & ---                                      & ---                                     & ---                          & ---  \\ \hline
			$21$  & $123{+}14$                   & $123{+}145{+}234{+}24{+}35{+}14$         & $145{+}234{+}24{+}35$                   & $123{+}145{+}24$             & $25$  \\ \hline
			$22^2$& $123{+}14{+}25$              & $123{+}145{+}23{+}25{+}34$               & $145{+}14{+}23{+}34$                    & $123{+}45$                   & $20$  \\ \hline
			$23$  & $123{+}145$                  & $123{+}145{+}245{+}24{+}35$              & $245{+}24{+}35$                         & $123{+}14$                   & $21$  \\ \hline
			$24$  & $123{+}145{+}23$             & $123{+}145{+}245{+}24{+}35$              & $245{+}23{+}24{+}35$                    & $123{+}14$                   & $21$  \\ \hline
			$25$  & $123{+}145{+}24$             & ---                                      & ---                                     & ---                          & ---  \\ \hline
			$26^3$& $123{+}145{+}23{+}24{+}35$   & $123{+}145{+}23{+}24{+}35$               & $0$                                     & $\leftarrow$                 & $0$  \\ \hline
			$27$  & $35$                         & $123{+}145{+}45{+}24{+}35$               & $123{+}145{+}45{+}24$                   & $123{+}145{+}23$             & $24$  \\ \hline
			$28^4$& $24{+}35$                    & $123{+}145{+}23{+}24{+}35$               & $123{+}145{+}23$                        & $\leftarrow$                 & $24$  \\ \hline
		\end{tabular}
	\end{adjustbox}
	\caption{Proof of Lemma 12 for the class $\widehat{\mathcal{R}}_{1,5}^3$.}\label{table:RM1526Cosets}
\end{table}

\begin{table}[]
	\begin{adjustbox}{width=\columnwidth,center}
		\begin{tabular}{|l|l|l|l|l|l|}
			\hline
			No    & Representative $f$           & $g$ from $\widehat{\mathcal{R}}_{1,5}^4 (28)$ & Sum $h=f{+}g$                & $h$ is equal to      & $C(h)$ \\ \hline
			
			$0$   & $0$                          & ---                                      & ---                          & ---                          & --- \\ \hline
			$1$   & $2345$                       & $12{+}34$                                & $2345{+}12{+}34$             & $\leftarrow$                 & $5$  \\ \hline
			$2$   & $2345{+}12$                  & $12{+}34$                                & $2345{+}34$                  & $2345{+}23$                  & $3$ \\ \hline
			$3$   & $2345{+}23$                  & ---                                      & ---                          & ---                          & ---  \\ \hline
			$4$   & $2345{+}23{+}45$             & $23{+}45$                                & $2345$                       & $\leftarrow$                 & $1$  \\ \hline
			$5$   & $2345{+}12{+}34$             & ---                                      & ---                          & ---                          & ---  \\ \hline
			$6$   & $2345{+}123$                 & $12{+}34$                                & $2345{+}123{+}12{+}34$       & $\leftarrow$                 & $11$  \\ \hline
			$7$   & $2345{+}123{+}12$            & $12{+}34$                                & $2345{+}123{+}34$            & $2345{+}123{+}24$            & $8$  \\ \hline
			$8$   & $2345{+}123{+}24$            & ---                                      & ---                          & ---                          & ---  \\ \hline
			$9$   & $2345{+}123{+}14$            & $14{+}35$                                & $2345{+}123{+}35$            & $2345{+}123{+}24$            & $8$  \\ \hline
			$10$  & $2345{+}123{+}45$            & $12{+}45$                                & $2345{+}123{+}12$            & $\leftarrow$                 & $7$  \\ \hline
			$11$  & $2345{+}123{+}12{+}34$       & ---                                      & ---                          & ---                          & ---  \\ \hline
			$12$  & $2345{+}123{+}14{+}35$       & $14{+}35$                                & $2345{+}123$                 & $\leftarrow$                 & $6$  \\ \hline
			$13$  & $2345{+}123{+}12{+}45$       & $12{+}45$                                & $2345{+}123$                 & $\leftarrow$                 & $6$  \\ \hline
			$14^1$& $2345{+}123{+}24{+}35$       & $24{+}35$                                & $2345{+}123$                 & $\leftarrow$                 & $6$  \\ \hline
			$15$  & $2345{+}123{+}145$           & $24{+}35$                                & $2345{+}123{+}145{+}24{+}35$ & $\leftarrow$                 & $18$  \\ \hline
			$16$  & $2345{+}123{+}145{+}45$      & $23{+}45$                                & $2345{+}123{+}145{+}23$      & $2345{+}123{+}145{+}45$      & $16$  \\ \hline
			$17$  & $2345{+}123{+}145{+}24{+}45$ & $23{+}45$                                & $2345{+}123{+}145{+}24{+}23$ & $2345{+}123{+}145{+}24{+}45$ & $17$  \\ \hline
			$18$  & $2345{+}123{+}145{+}24{+}35$ & ---                                      & ---                          & ---                          & ---  \\ \hline
			$19$  & $123$                        & $12{+}45$                                & $123{+}12{+}45$              & $123{+}45$                   & $20$  \\ \hline
			$20$  & $123{+}45$                   & ---                                      & ---                          & ---                          & ---  \\ \hline
			$21$  & $123{+}14$                   & $14{+}25$                                & $123{+}25$                   & $123{+}14$                   & $21$  \\ \hline
			$22^2$& $123{+}14{+}25$              & $14{+}25$                                & $123$                        & $\leftarrow$                 & $19$  \\ \hline
			$23$  & $123{+}145$                  & $23{+}45$                                & $123{+}145{+}23{+}45$        & $123{+}145$                  & $23$  \\ \hline
			$24$  & $123{+}145{+}23$             & $23{+}45$                                & $123{+}145{+}45$             & $123{+}145{+}23$             & $24$  \\ \hline
			$25$  & $123{+}145{+}24$             & $24{+}35$                                & $123{+}145{+}35$             & $123{+}145{+}24$             & $25$  \\ \hline
			$26^3$& $123{+}145{+}23{+}24{+}35$   & $24{+}35$                                & $123{+}145{+}23$             & $\leftarrow$                 & $24$  \\ \hline
			$27$  & $12$                         & $12{+}34$                                & $34$                         & $12$                         & $27$  \\ \hline
			$28^4$& $12{+}34$                    & $12{+}34$                                & $0$                          & $\leftarrow$                 & $0$  \\ \hline
		\end{tabular}
	\end{adjustbox}
	\caption{Proof of Lemma 12 for the class $\widehat{\mathcal{R}}_{1,5}^4$.}\label{table:RM1528Cosets}
\end{table}

\end{document}